\documentclass[12pt, reqno]{amsart}
\usepackage{amssymb}
\usepackage{amsmath}
\usepackage{graphicx,color}
\usepackage{wrapfig,framed}
\usepackage{mathtools}
\usepackage[height=22cm, width=15.7cm, hmarginratio={1:1}]{geometry}
\usepackage{hyperref}

\theoremstyle{plain}
\newtheorem{theorem}{Theorem}
\newtheorem{prop}{Proposition}
\newtheorem{lemma}{Lemma}
\newtheorem{cor}{Corollary}
\theoremstyle{definition}
\newtheorem{defi}{Definition}
\newtheorem{remark}{Remark}
\newtheorem{example}{Example}

\newcommand{\beq}{\begin{equation}}
\newcommand{\eeq}{\end{equation}}
\newcommand{\nn}{\nonumber}

\newcommand{\bllb}{\bigl(\hskip -0.05truecm \bigl(}
\newcommand{\brrb}{\bigr)\hskip -0.05truecm \bigr)}

\newcommand{\bllm}{\bigl[\hskip -0.05truecm \bigl[}
\newcommand{\brrm}{\bigr]\hskip -0.05truecm \bigr]}

\newcommand{\cS}{\mathcal{S}}
\newcommand{\cL}{\mathcal{L}}
\newcommand{\cB}{\mathcal{B}}
\newcommand{\QQ}{{\mathbb Q}}
\newcommand{\CC}{{\mathbb C}}
\newcommand{\ZZ}{{\mathbb Z}}

\newcommand{\A}{{\mathcal A}}

\newcommand{\tr}{{\rm tr}}
\newcommand{\Tr}{{\rm Tr}}
\newcommand{\bt}{{\bf t}}
\newcommand{\e}{\epsilon}
\newcommand{\bdzero}{{\bf 0}}

\newcommand{\p}{\partial}
\newcommand{\epf}{$\quad$\hfill\raisebox{0.11truecm}{\fbox{}}\par\vskip0.4truecm}

\def\={\; = \;}
\def\+{\; + \;}
\def\:={\, := \, }

\def \wt {\widetilde}

\begin{document}
\title[Tau-functions for Toda lattice hierarchy]{On tau-functions for the Toda lattice hierarchy}
\author{Di Yang}
\dedicatory{Dedicated to the memory of Boris Anatol'evich Dubrovin, with gratitude and admiration}
\begin{abstract}
We extend a recent result of \cite{DYZ2} for the KdV hierarchy
to the Toda lattice hierarchy. Namely, 
for an arbitrary solution to the Toda lattice hierarchy, 
we define a pair of wave functions, 
and use them to give explicit formulae for the generating 
series of $k$-point correlation functions 
of the solution. Applications to computing GUE correlators and 
Gromov--Witten invariants of the Riemann sphere are under consideration.
\end{abstract}
\maketitle

\setcounter{tocdepth}{1}
\tableofcontents

\section{Introduction} 
The Toda lattice hierarchy, which contains  
the Toda lattice equation
\beq
\ddot \sigma(n) \= e^{\sigma(n-1)-\sigma(n)}  \, - \, e^{\sigma(n)-\sigma(n+1)} \,,   \label{Todalatticeeqn}
\eeq
is an important {\it integrable hierarchy} of nonlinear 
differential-difference equations \cite{FT,fla,mana,UT}.
In this paper, following the idea of~\cite{DYZ2} we derive new formulae for generating series 
of $k$-point correlation functions for the Toda lattice hierarchy by using the 
matrix resolvent approach~\cite{DuY1} and by introducing {\it a pair of wave functions}.
\subsection{Toda lattice hierarchy and tau-function}
Let 
\beq
\A \:= \ZZ \, [ v_0,w_0,v_{\pm 1},  w_{\pm 1}, v_{\pm 2}, w_{\pm 2}, \cdots ]
\eeq
be the polynomial ring. Define the shift operator  $\Lambda:\A \rightarrow \A$ via 
\[
\Lambda(1)\=1\,, \quad \Lambda (v_i) \= v_{i+1}\,, \quad  \Lambda (w_i) \= w_{i+1} \,,  
\quad \Lambda (f g) \= \Lambda (f) \, \Lambda (g)\]
$\forall \, i\in \ZZ$ and $f,g\in \A$.
Denote by~$\Lambda^{-1}$ the inverse of~$\Lambda$ satisfying 
$\Lambda^{-1} (v_i) = v_{i-1}$, $\Lambda^{-1} (w_i) = w_{i-1}$, and 
$\Lambda^{-1} (f g) = \Lambda^{-1} (f) \, \Lambda^{-1} (g)$.
For a difference operator~$P$ on~$\A$, we mean an operator of the form 
$P = \sum_{m\in \ZZ} P_m \, \Lambda^m $,
where $P_m \in \A$. 
Denote $P_+:=\sum_{m\geq 0} P_m \, \Lambda^m$, $P_-:=\sum_{m< 0} P_m \, \Lambda^m$, ${\rm Coef}(P,m):=P_m$. 
A linear operator 
$D:\A \to \A$ is called a derivation on~$\A$, if
$$
D(fg) \= D(f) \, g \+ f \, D(g) \,, \quad \forall\,f,g\in \A \, . 
$$
The derivation~$D$ is called {\it admissible} if it commutes with~$\Lambda$.
Clearly, every admissible derivation~$D$ is uniquely determined by the values~$D(v_0)$ and~$D(w_0)$. 
Let 
\beq \label{LaxL} 
L\:=\Lambda \+ v_0 \+ w_0 \, \Lambda^{-1}
\eeq
be a difference operator, and define a sequence of difference operators $A_k$,~$k\geq 0$ by
\beq
A_k \:= \bigl(L^{k+1}\bigr)_+ \,.
\eeq
We associate with~$A_k$ a sequence of admissible derivations $D_k:\A \to \A$ defined via
\beq\label{todaderiv}
D_k(v_0) \:= {\rm Coef} \bigl([A_k,L], 0\bigr) \,, \quad 
D_k(w_0) \:= {\rm Coef} \bigl([A_k,L], -1\bigr) \,, \qquad k\geq0\,.
\eeq
The first few $D_k(v_0)$ and $D_k(w_0)$ are $D_0(v_0)=w_1-w_0$, $D_0(w_0)=w_0 \, (v_0-v_{-1})$; $D_1(v_0)=w_1(v_1+v_0)-w_0(v_0+v_{-1})$, $D_1(w_0)=w_0\bigl(w_1-w_{-1}+v_0^2-v_{-1}^2\bigr)$; etc.. 

\begin{lemma} \label{commtoda}
The operators~$D_k$, $k\geq 0$ pairwise commute.
\end{lemma}
\noindent  
This lemma was known. We call $D_k$ the Toda lattice derivations, and 
\eqref{todaderiv} the abstract Toda lattice hierarchy.

A {\it tau-structure} associated to the derivations~$(D_k)_{k\geq 0}$ is a collection of polynomials 
$\bigl(\Omega_{p,q}, S_p\bigr)_{p,q\geq 0}$ in~$\mathcal{A}$ satisfying
\begin{align}
& \Omega_{p,q} \= \Omega_{q,p} \,,   \quad    
 D_r  \bigl(\Omega_{p,q}\bigr) \=  D_q \bigl(\Omega_{p,r}\bigr) \,,   \label{taustructure}  \\
&  (\Lambda -1 ) \, \bigl(\Omega_{p,q}\bigr) \= D_q \bigl(S_p\bigr)  \,, \label{taustructure2} \\ 
& w_0 \, \bigl(1 \,-\, \Lambda^{-1} \bigr) \bigl(S_p\bigr) \= D_p (w_0)  \,\label{taustructure3} 
\end{align}
for all $p,q,r\geq 0$.
It can be shown (e.g.~\cite{DuY1}) that the tau-structure exists
and is unique up to replacing $\Omega_{p,q},S_p$ by $\Omega_{p,q}+c_{p,q}$ and $S_p+a_p$
respectively, where $c_{p,q}=c_{q,p}$ and~$a_ p$ are arbitrary constants.  
The tau-structure $\Omega_{p,q},S_p$ is called canonical if
$$
\Omega_{p,q}\big|_{v_i=0, \, w_i=0, \, i\in \ZZ} \= 0 \,, \quad S_{p} \big |_{v_i=0, \, w_i=0, \, i\in \ZZ} \= 0 \,. 
$$

Let us take $\Omega_{p,q},S_p$ the canonical tau-structure. For $m\geq 3$, define
\beq\label{symmetryindices}
\Omega_{p_1,\dots,p_m} \:= D_{p_1} \cdots D_{p_{m-2}} \,
\bigl(\Omega_{p_{m-1} p_m} \bigr) \in \A \,, \qquad p_1,\dots,p_m\geq 0\,.
\eeq
By~\eqref{taustructure} we know that the $\Omega_{p_1,\dots,p_m}$, $m\geq 2$ are totally symmetric with respect to permutations of the indices $p_1,\dots,p_m$.  
The first few of these polynomials are
\begin{align}
& S_0 \= v_0 \,, \quad S_1 \= w_1+w_0 +v_0^2 \,, \\
& \Omega_{0,0} \= w_0 \,, \quad \Omega_{0,1}\=\Omega_{1,0}\=w_1(v_1+v_0)\,. 
\end{align}
 
If we think of~$v_0$, $w_0$ as two functions $v(n)$, $w(n)$ of~$n$, respectively, and 
$v_i$, $w_i$ as $v(n+i)$, $w(n+i)$, then 
the Toda lattice derivations~$D_k$ lead to a hierarchy of evolutionary 
differential-difference equations, called the Toda lattice hierarchy, given by 
\begin{align}
& \frac{\p v(n)}{\p t_k} \=  D_k(v_0) (n)  \,, \qquad \frac{\p w(n)}{\p t_k} \= D_k(w_0) (n) \,, \label{Todauw} 
\end{align}
where $k\geq 0$, and the $D_k(v_0) (n), D_k(w_0) (n)$ 
are defined as $D_k(v_0), D_k(w_0)$ with $v_i$, $w_i$ replaced by $v(n+i)$, $w(n+i)$, respectively.
Lemma~\ref{commtoda} implies that the flows \eqref{Todauw} all
commute. So we can solve the whole Toda lattice hierarchy \eqref{Todauw} 
together, which yields solutions of the form $(v=v(n,\bt),w=w(n,\bt))$. Here  
 $\bt:=(t_0,t_1,\dots)$ denotes the infinite time vector. Note that the $k=0$ equations read 
\begin{align}
& \dot v(n) \=  w(n+1) \,- \, w(n)\,, \qquad \dot w(n) \= w(n) \, \bigl(v(n)-v(n-1)\bigr) \,, \label{Todaequw} 
\end{align}
which are equivalent to equation~\eqref{Todalatticeeqn} via the transformation 
$$
w(n) \= e^{\sigma(n-1)-\sigma(n)} \,, \qquad v(n) \= - \dot \sigma(n) \,. 
$$
Here, dot, ``\,$\dot{}$\,", is identified with $\p/\p t_0$.

Let~$V$ be a ring of functions of~$n$ closed under shifting~$n$ by~$\pm1$. 
For two given $f(n),g(n)\in V$, consider the 
initial value problem for~\eqref{Todauw} with the initial condition:
\beq \label{ini} 
v(n,\bdzero) \=f(n)\,, \quad  w(n,\bdzero)\=g(n) \,. 
\eeq 
The solution $(v(n,\bt), w(n,\bt))\in V[[\bt]]^2$ exists and is unique, which gives the following 1-1 correspondence: 
\beq\label{11correspondence}
\mbox{$\bigl\{$solution $(v,w)$ of \eqref{Todauw}  in $V[[\bt]]^2\bigr\}$ \;  
$\longleftrightarrow$ \; $\bigl\{$initial data $ (f,g) \bigr\}$}\, .
\eeq

\begin{example}  $f(n)=0,~g(n)=n$. (For this case, one 
can take $V=\QQ[n]$.)
The corresponding unique solution governs the enumerations of ribbon graphs in all genera.  
\end{example}

\begin{example} $f(n)= (n+\frac12)\e, ~ g(n)=1$. (For this case, one 
can take $V=\QQ[\e][n]$.) The corresponding unique solution governs the 
Gromov-Witten invariants of~$\mathbb{P}^1$ in the stationary sector in all genera and all degrees. 
\end{example}

Let $(v,w)\in V[[\bt]]^2$ be an arbitrary solution to the 
Toda lattice hierarchy~\eqref{Todauw}. 
Write $\Omega_{p,q}(n,\bt)$ and~$S_p(n,\bt)$ 
as the images of $\Omega_{p,q}$ and~$S_p$ 
under the substitutions 
\beq \label{substitutionsuw}
v_i\mapsto v(n+i,\bt),\quad w_i\mapsto w(n+i,\bt),\qquad i\in\ZZ,
\eeq 
respectively. (Similar notations will be used for other elements of~$\A$.)
Equalities~\eqref{taustructure} then imply the existence of a function $\tau=\tau(n,\bt)$ such that for $p,q\geq0$,   
\begin{align} 
& \Omega_{p,q}(n,\bt) \= \frac{\p^2 \log \tau(n,\bt)}{\p t_p \p t_q } \,,\label{taufunctionOmega} \\
& S_p(n,\bt) \=  \frac{\p}{\p t_p}\log \frac{\tau(n+1,\bt)}{\tau(n,\bt)} \,, \label{taufunctionS} \\ 
& w(n,\bt)  \= \frac{\tau(n+1,\bt) \, \tau(n-1,\bt)}{\tau(n,\bt)^2} \,. \label{taufunctionw} 
\end{align}
We call $\tau(n,\bt)$ the {\it Dubrovin--Zhang (DZ) type tau-function} \cite{DZ-norm,DuY1} of the solution~$(v,w)$, 
for short the tau-function of the solution.
The symmetry in~\eqref{symmetryindices} 
is more obvious:  the image $\Omega_{p_1,\dots,p_m}(n,\bt)$ of $\Omega_{p_1,\dots,p_m}$ under~\eqref{substitutionsuw} satisfies
\beq
\Omega_{p_1,\dots,p_m}(n,\bt) \= \frac{\p^m \log \tau(n,\bt)}{\p t_{p_1} \cdots \p t_{p_m}}, \; \qquad m\geq 2,\; p_1,\dots,p_m\geq 0\,. \label{Taylorlogtau}
\eeq
Define $\Omega_p(n,\bt)=\p_{t_p} \log \tau(n,\bt)$, $p\geq 0$. 
These logarithmic derivatives of~$\tau(n,\bt)$ are 
called {\it correlation functions} of the solution~$(v,w)$.
The specializations 
$\Omega_{p_1,\dots,p_m}(n,\bdzero)$ are 
called {\it $m$-point partial correlation functions} of~$(v,w)$.
\begin{remark} The tau-function~$\tau(n,\bt)$ of the solution~$(v,w)$ is unique up to multiplying it by the exponential of a linear function of $n,t_0,t_1,t_2,\cdots$. 
\end{remark}
\subsection{Matrix resolvent} The matrix resolvent (MR) method
 for computing correlation functions for integrable hierarchies was introduced in~\cite{BDY1,BDY2,BDY3}, 
 and was extended to the discrete case in~\cite{DuY1} (in particular to the Toda lattice hierarchy).
Denote 
\[
U(\lambda)\:=\begin{pmatrix} v_0-\lambda & w_0\\ -1 & 0\end{pmatrix} \,. 
\]
The following lemma for the Toda lattice hierarchy was proven in~\cite{DuY1}.
\begin{lemma}[\cite{DuY1}] \label{mrlemma}
There exists a unique series $R(\lambda)\in {\rm Mat}\bigl(2,\A[[\lambda^{-1}]]\bigr)$ satisfying
\begin{align}
& \Lambda \bigl(R(\lambda)\bigr) \, U(\lambda) \,-\,U(\lambda) \, R(\lambda) \= 0 \,, \label{defbr1}\\
& {\rm Tr} \, R (\lambda) \= 1\,, \quad \det R(\lambda) \= 0 \,, \label{defbr2}\\
& R(\lambda) \,-\, \begin{pmatrix} 1 & 0 \\ 0 & 0 \end{pmatrix} 
\; \in \; {\rm Mat}\bigl(2,\A[[\lambda^{-1}]]\lambda^{-1}\bigr)\,. \label{defbr3}
\end{align}
\end{lemma}
\noindent The unique series $R(\lambda)$ in Lemma~\ref{mrlemma} 
is called the {\it basic matrix resolvent}.  The first few terms of~$R(\lambda)$ are given by
\begin{align}
& R(\lambda)\= \begin{pmatrix} 1 & 0 \\ 0 & 0\end{pmatrix} \+ \begin{pmatrix} 0 & -w_0 \\ 1 & 0\end{pmatrix} \frac1{\lambda} 
\+ \begin{pmatrix} w_0 & -v_0w_0 \\ v_{-1} & -w_0\end{pmatrix} \frac1{\lambda^2} \nn\\
& \qquad \qquad \qquad \+  \begin{pmatrix} w_0(v_0+v_{-1}) & -w_0(w_0+w_1+v_0^2) \\ w_0+w_{-1}+v_{-1}^2 & -w_0(v_0+v_{-1})\end{pmatrix} \frac1{\lambda^3}  \+\cdots\,. \label{Rfirstfew}
\end{align}
\begin{prop}[\cite{DuY1}] \label{matrixprop}
For any $k\geq 2$, the following formula holds true:
\begin{align}
\sum_{i_1,\dots,i_k\geq 0} 
\frac{\Omega_{i_1,\dots,i_k}}{\lambda_1^{i_1+2} \cdots \lambda_k^{i_k+2}} \= 
- \sum_{\pi\in \mathcal{S}_k/C_k} \frac{\tr \, \prod_{j=1}^k R\bigl(\lambda_{\pi(j)}\bigr)}
{\prod_{j=1}^k \bigl(\lambda_{\pi{(j)}}-\lambda_{\pi{(j+1)}}\bigr)} 
\,-\, \frac{\delta_{k,2}}{(\lambda_1-\lambda_2)^2} \,, \label{propgueproved1}
\end{align}
where $\mathcal{S}_k$ denotes the symmetry group and $C_k$ the cyclic group, and $\pi(k+1):=\pi(1)$. 
\end{prop}

\noindent The meaning of~\eqref{propgueproved1} is the following: For any fixed permutation $(j_1,\dots,j_k)$ of $(1,\dots,k)$,
 expanding the right-hand side with respect to $|\lambda_{j_1}|>\dots>|\lambda_{j_k}|>>0$ 
gives identical formal power series with the left-hand side. This is because, after the summation 
over the $\cS_k/C_k$ and subtracting $\frac{\delta_{k,2}}{(\lambda_1-\lambda_2)^2}$, the poles in the diagonal cancel 
(cf.~the Proposition~2 of~\cite{DYZ1} for a straightforward proof of this point). 
We note that, as formal power series, the coefficients of the both sides of~\eqref{propgueproved1} are in 
$\A$. We give in Section~\ref{section2} a new proof of~\eqref{propgueproved1}, where we keep all 
derivations with coefficients in~$\A$.

\subsection{From wave functions to correlation functions}\label{Introwave}
In~\cite{DYZ2} we introduced the notion of a tuple of 
wave functions (in many cases {\it a pair}) to the 
study of tau-function without using the Sato theory. Let us generalize it 
to the Toda lattice hierarchy. Our definition of a pair will be based on the standard construction of 
wave functions for the Toda lattice hierarchy~\cite{UT,CDZ,Carlet}.
For given $(f(n),g(n))$ a pair of arbitrary elements in~$V$, let $L$ be the 
linear difference operator $L=\Lambda+f(n)+g(n) \Lambda^{-1}$. 
Denote 
\beq\label{inir}
s(n) \:= -  \bigl(1-\Lambda^{-1}\bigr)^{-1} \bigl(\log g(n)\bigr)\,.
\eeq
The function $s(n)$ is in a certain extension
~$\widehat V$ of~$V$, 
and is uniquely determined by~$\log g(n)$ up to a constant. 
Below we fix a choice of~$s(n)$.
An element $\psi_A(\lambda,n)=\bigl(1+{\rm O}(\lambda^{-1})\bigr) \,\lambda^n$ in the module 
$\widetilde V[[ \lambda^{-1} ]]  \lambda^n$
is called a (formal) wave function of type~A associated to~$f(n),g(n)$, if
$L \bigl( \psi_A(\lambda,n)\bigr)= \lambda \psi_A(\lambda,n)$. Here, 
$\wt V$ is a ring of functions of~$n$ satisfying 
\[V\subset (\Lambda -1)\bigl(\widetilde V\bigr) \subset \widetilde V\,.\]
An element $\psi_B(\lambda,n)=\bigl(1+{\rm O}(\lambda^{-1})\bigr) \, e^{-s(n)} \lambda^{-n}$ 
in the module $\widetilde V [[ \lambda^{-1} ]]  e^{-s(n)} \lambda^{-n}$ 
is called a (formal) wave function of type~$B$, if
$L\bigl( \psi_B(\lambda,n) \bigr) = \lambda \psi_B(\lambda,n) $.
Let 
$
\psi_A \in \widetilde V[[ \lambda^{-1}]] \, \lambda^n$ 
and 
$\psi_B \in \widetilde V [[ \lambda^{-1} ]]  e^{-s(n)} \lambda^{-n}$
be two wave functions of type A and of type B associated to $f(n),g(n)$, respectively. 
Define 
\beq\label{definitiondtimeind}
d(\lambda,n) \:= \psi_A(\lambda,n) \, \psi_B(\lambda,n-1) - \psi_B(\lambda,n) \, \psi_A(\lambda,n-1) \,. 
\eeq
We call $\psi_A,\psi_B$ form {\it a pair} if the following normalization condition holds:
\beq\label{normalization}
e^{s(n-1)} d(\lambda,n) \= \lambda \,. 
\eeq
The existence of a pair of wave functions is proven in Section~\ref{section3}.

Denote by $(v(n,\bt),w(n,\bt))$ the unique solution in~$V[[\bt]]^2$ to 
the Toda lattice hierarchy with $(f(n),g(n))$ as its initial value, 
by $\psi_A(\lambda,n)$ and~$\psi_B(\lambda,n)$ a pair of wave functions associated to $(f(n),g(n))$, and 
by $\tau(n,\bt)$ the DZ type tau-function of $(v(n,\bt),w(n,\bt))$. Introduce 
\beq\label{defineDini}
D(\lambda,\mu,n) \:= \frac{\psi_A(\lambda,n) \, \psi_B(\mu,n-1) \,-\, \psi_A(\lambda,n-1) \,\psi_B(\mu,n)}{\lambda-\mu} \,.
\eeq

\begin{theorem}\label{main1} 
Fix $k\geq 2$ being an integer. 
The generating series of $k$-point partial correlation functions has the following expression:
\begin{align}
& \sum_{i_1,\dots,i_k\geq 0} 
\frac{\p^k \log \tau}{\p t_{i_1} \dots \p t_{i_k}} (n,\bdzero) \, \frac1{\lambda_1^{i_1+2} \cdots \lambda_k^{i_k+2}}  \nn\\
& \qquad \qquad  \=  
(-1)^{k-1} \frac{e^{k s(n-1)}}{\prod_{j=1}^k \lambda_{j}}  
\sum_{\pi \in \mathcal{S}_k/C_k} \prod_{j=1}^k D(\lambda_{\pi(j)},\lambda_{\pi(j+1)},n) 
\,-\, \frac{\delta_{k,2}}{(\lambda_1-\lambda_2)^2}  \,.  \label{mainidentityfg}
\end{align}
\end{theorem}

Theorem~\ref{main1} gives an algorithm with the initial value $(f(n),g(n))$ as the only input for computing  
the $k_{\rm th}$-order logarithmic derivatives of the tau-function~$\tau(n,\bt)$ evaluated at $\bt=\bdzero$ 
for $k\geq 2$. Indeed, by solving the spectral problem $L(\psi)=\lambda \psi$ with $L=\Lambda+f(n)+g(n)\Lambda^{-1}$ 
and with the normalization condition~\eqref{normalization}, one constructs a pair of wave functions; 
the coefficients in the ${\bt}$-expansion of~$\log\tau(n,\bt)$ 
are then obtained through algebraic manipulations by using~\eqref{mainidentity}. 
(Recall that in the inverse scattering method (cf. e.g.~\cite{FT,fla}),   
an additional integral equation needs to be solved.) 
Two applications of Theorem~\ref{main1} are given in Section~\ref{section5}.
For a certain class of bispectral solutions (cf.~\cite{GY}) 
it would be possible to give a {\it canonical} way of 
constructing a pair of wave functions, which  
was briefly mentioned in~\cite{DYZ2} for the KdV hierarchy; we plan to do this for KdV and for Toda lattice in a future 
publication.

\medskip

\noindent \textbf{Organization of the paper.} 
In Section~\ref{section2} we review the MR method of studying tau-structure 
for the Toda lattice hierarchy. 
In Section~\ref{section3} we prove the existence of a pair of wave functions. 
In Section~\ref{section4} we prove Theorem~\ref{main1} and several other theorems. 
Applications to the computations of GUE correlators and 
Gromov--Witten invariants of~$\mathbb{P}^1$ are given in Section~\ref{section5}.
In Appendix~A we give an extension of~$\A$, define a pair of abstract pre-wave functions, and 
prove an abstract version for Theorem~\ref{main1}.

\medskip

\noindent \textbf{Acknowledgements.} 
The author is grateful to Youjin Zhang, Boris Dubrovin, Don Zagier for their advising, 
 and to Jian Zhou and Si-Qi Liu for helpful discussions. 
He thanks the referee for valuable 
suggestions; in particular, Appendix~A comes out from the suggestions. 
He also wishes to thank Boris Dubrovin for introducing GUE to him 
and for helpful suggestions and discussions on this article.
The work is partially supported by a starting research grant from 
University of Science and Technology of China.

\section{Matrix resolvent and tau-structure}\label{section2}
We continue in this section with more details in reviewing the   
MR method~\cite{DuY1} to the Toda lattice hierarchy. 
Denote by~$\cL$ the matrix Lax operator for the Toda lattice: 
\[ \cL \:=  \begin{pmatrix} \Lambda & 0\\ 0 & \Lambda\end{pmatrix} 
\+ \begin{pmatrix} v_0-\lambda & w_0\\ -1 & 0\end{pmatrix} \=
\Lambda \+ U(\lambda) \,. \]
Let $R(\lambda)$ be the basic matrix resolvent (of~$\cL$). 
Write
\begin{align}
& R(\lambda) \= \begin{pmatrix} 
1+\alpha(\lambda) & \beta(\lambda) \\ \gamma(\lambda) & -\alpha(\lambda) \end{pmatrix} \,,  \\
& \alpha(\lambda) \= \sum_{i\geq 0} \frac{a_i}{\lambda^{i+1}} \,, \quad 
\beta(\lambda) \= \sum_{i\geq 0} \frac{b_i}{\lambda^{i+1}} \,,  \quad 
\gamma(\lambda) \= \sum_{i\geq 0} \frac{c_i}{\lambda^{i+1}} \,, 
\end{align}
where $a_i,b_i,c_i\in \A$. 
From 
the defining equations~\eqref{defbr1}--\eqref{defbr3}, 
we see that the series $\alpha,\beta,\gamma$ satisfy the equations
\begin{align}
& \beta(\lambda) \= -w_0 \, \Lambda \bigl(\gamma(\lambda)\bigr) \,,  \label{rr1} \\
& \gamma(\lambda) \= \frac{1\+\alpha(\lambda) \+ \Lambda^{-1} \bigl(\alpha(\lambda)\bigr)}
{\lambda - v_{-1}} \,,  \label{rr2} \\
& \Bigl(\alpha(\lambda)- \Lambda \bigl(\alpha(\lambda)\bigr)\Bigr)(\lambda-v_0) \,-\, w_0 \, \frac{1+\alpha(\lambda)+\Lambda^{-1}\bigl(\alpha(\lambda)\bigr)}{\lambda-v_{-1}} \nn\\
& \qquad\qquad\qquad\qquad\qquad\quad  \+ w_1 \,
\frac{1+\Lambda\bigl(\alpha(\lambda)\bigr)+\Lambda^2\bigl(\alpha(\lambda)\bigr)}{\lambda-v_1} \=0\,, \label{rr3alpha}\\
& \alpha(\lambda)+\alpha(\lambda)^2 \+ \beta(\lambda) \gamma(\lambda) \= 0\,.
\end{align}
These equalities give rise to the following recursion relation for $a_i,b_i,c_i$:
\begin{align}
& b_j\= -w_0 \, \Lambda (c_{j})\,, \qquad c_{j+1} \= v_{-1} \, c_{j} \+ \bigl(1 \+ \Lambda^{-1}\bigr) \, (a_j)\,,\label{rec1}\\
& \bigl(1-\Lambda \bigr) \, (a_{j+1}) \+ v_0 \, \bigl(\Lambda-1\bigr) \, (a_j)  \+ w_1 \, \Lambda^2 (c_{j}) \,-\, w_0 \, c_{j} \= 0 \,, \label{rec2} \\
& a_{\ell}=
\sum_{i+j=\ell-1} \Bigl(w_0 \, c_{i} \, \Lambda (c_j) \,-\, a_{i} \, a_{j}  \Bigr) \label{rec3}
\end{align}
along with  
\beq\label{iniac}
a_0\=0 \,, \qquad c_0\=1\,.
\eeq
Equations~\eqref{rec1}--\eqref{iniac} are called the matrix-resolvent recursion relation.

It has been proven~\cite{DuY1} that the abstract Toda lattice hierarchy~\eqref{todaderiv} 
can be equivalently written as
\begin{align}
& D_j  \, (v_0) \= \bigl(\Lambda-1\bigr) \, (a_{j+1}) \,, \nn\\
& D_j \, (w_0) \= w_0 \, \bigl(\Lambda-1\bigr)\,(c_{j+1})\,, \nn
\end{align}
where $j\geq 0$. Define an operator~$\nabla(\lambda)$ by
\beq\label{nabladef}
\nabla(\lambda) \:= \sum_{j\geq 0} \frac{D_j}{\lambda^{j+2}} \,.
\eeq
We have 
\begin{align}
& \nabla(\lambda) \, ( v_0 ) \=   \bigl(\Lambda-1\bigr) \, \bigl(\alpha(\lambda)\bigr)  \,, \\
& \nabla(\lambda) \, ( w_0 ) \= w_0 \, \bigl(\Lambda-1\bigr) \, \bigl(\gamma(\lambda)-1\bigr) \,.
\end{align}

\begin{lemma}\label{observationDR}
There exists a unique element $W(\lambda,\mu)$ in $\A\otimes {\rm sl}_2(\CC) [[ \lambda^{-1},\mu^{-1}]] \lambda^{-1} \mu^{-1}$ 
of the form
$$
W(\lambda,\mu) \= \begin{pmatrix} 
X(\lambda,\mu) & Y(\lambda,\mu) \\ Z(\lambda,\mu) & -X(\lambda,\mu)\end{pmatrix}
$$ 
satisfying the following linear inhomogeneous equations for the entries of~$W$:
\begin{align}
& \Lambda \bigl(W(\lambda,\mu)\bigr)  \, U(\lambda)  \,-\, U(\lambda) \, W(\lambda,\mu) 
\+ \Lambda \bigl(R(\lambda)\bigr) \, \nabla(\mu) \bigl(U(\lambda)\bigr) \,-\,
\nabla(\mu) \bigl(U(\lambda)\bigr) \, R(\lambda)   \=0 \,, \label{W1} \\
& X(\lambda,\mu) \+  2  \alpha(\lambda) \, X(\lambda,\mu) 
\+ \gamma(\lambda) \, Y(\lambda,\mu) \+ \beta(\lambda) \, Z(\lambda,\mu) \= 0 \,.  \label{W2}
\end{align}
\end{lemma}
\begin{proof}
The existence part of this lemma follows from Lemma~\ref{mrlemma}. Indeed, if we define 
$$
W(\lambda,\mu) \:= \nabla(\mu) \bigl(R(\lambda)\bigr)\,, 
$$
then $W(\lambda,\mu)$ satisfies \eqref{W1}--\eqref{W2}. To see the uniqueness part, we first note that 
 the (1,2)-entry and the (2,1)-entry of the matrix equation~\eqref{W1} imply that 
$Y$ and $Z$ can be uniquely expressed in terms of~$X$. Indeed, we have
\begin{align}
& Z(\lambda,\mu) 
\= \tfrac{(1 + \Lambda^{-1}) (X(\lambda,\mu))}
{\lambda - v_{-1}}  \+ \gamma(\lambda) \tfrac{\Lambda^{-1} \circ \nabla(\mu) (v_0)}{\lambda - v_{-1}} \,, \\
& Y(\lambda,\mu) \= - \nabla(\mu) (w_0) \, \tfrac{1\+\alpha(\lambda) \+ \Lambda (\alpha(\lambda)) }
{\lambda - v_{0}}
\,- \, w_0 \, \tfrac{(1+ \Lambda) (X(\lambda,\mu))}
{\lambda - v_{0}}  \+ \beta(\lambda) \tfrac{\nabla(\mu) (v_0)}{\lambda - v_{0}}  \,.   
\end{align}
Substituting these two expressions in~\eqref{W2} we obtain the following linear inhomogeneous 
difference equation for~$X$:
\begin{align}
& \Bigl(1 +  2 \alpha(\lambda) +  \tfrac{\beta(\lambda)}{\lambda - v_{-1}}  -  \tfrac{w_0 \gamma(\lambda)}
{\lambda - v_{0}}  \Bigr) X(\lambda,\mu)   \,- \, \tfrac{ w_0 \gamma(\lambda) }
{\lambda - v_{0}}  \Lambda \bigl(X(\lambda,\mu)\bigr)
 \+ \tfrac{\beta(\lambda)}{\lambda - v_{-1}}  \Lambda^{-1} \bigl(X(\lambda,\mu)\bigr) \nn\\
 &  \quad \=   \Bigl(1+\alpha(\lambda) + \Lambda \bigl(\alpha(\lambda)\bigr)\Bigr) \gamma(\lambda) 
\tfrac{   \nabla(\mu) (w_0)  }
{\lambda - v_{0}}
  \,-\, \beta(\lambda) \gamma(\lambda) \bigl(1+\Lambda^{-1}\bigr) \Bigl(\tfrac{\nabla(\mu) (v_0)}{\lambda - v_{0}} \Bigr) \,.
\end{align}
Suppose this equation has two solutions $X_1,X_2$ in $\A [[ \lambda^{-1},\mu^{-1}]] \lambda^{-1} \mu^{-1}$.
Let $X_0=X_1-X_2$, then $X_0\in \A [[ \lambda^{-1},\mu^{-1}]] \lambda^{-1} \mu^{-1}$, 
and it satisfies the following equation:
\begin{align}
& \Bigl(1 +  2 \alpha(\lambda) +  \tfrac{\beta(\lambda)}{\lambda - v_{-1}}  -  \tfrac{w_0 \gamma(\lambda)}
{\lambda - v_{0}}  \Bigr) X_0(\lambda,\mu)   \,- \, \tfrac{ w_0 \gamma(\lambda) }
{\lambda - v_{0}}  \Lambda \bigl(X_0(\lambda,\mu)\bigr) \nn\\
& \qquad\qquad\qquad\qquad\qquad\qquad\qquad \+ \tfrac{\beta(\lambda)}{\lambda - v_{-1}} \, \Lambda^{-1} \bigl(X_0(\lambda,\mu)\bigr)  \=   0\,.
\end{align}
It follows that $X_0$ vanishes. Indeed, write 
$X_0=\sum_{j\geq 0} X_{0,j}(\mu) \lambda^{-(j+1)}$. Observe that 
\[
\tfrac1{\lambda-v_{m}} \= \tfrac1\lambda \+ \tfrac{v_m}{\lambda^2} \+ \cdots \; \in \; \A [[ \lambda^{-1}]] \lambda^{-1} \,, \qquad m=-1,0\,, 
\]
and recall that $\alpha(\lambda),\beta(\lambda),\gamma(\lambda)\in\A[[\lambda^{-1}]]\lambda^{-1}$. Then by comparing 
the coefficients of powers of~$\lambda^{-1}$ consecutively we find that $X_{0,0}(\mu)=0$, $X_{0,1}(\mu)=0$, $X_{0,2}(\mu)=0$, $\cdots$. 
So $X_0=0$. Hence $X_1=X_2$. The lemma is proved.  
\end{proof}

Based on this lemma we now give a new proof for 
the following proposition. 
\begin{prop}[\cite{DuY1}] \label{thelemma}
The following equation holds true:
\beq\label{naR}
\nabla(\mu) \, R(\lambda) \= \frac1{\mu-\lambda} \bigl[ R(\mu), R(\lambda) \bigr] + \bigl[Q(\mu),R(\lambda)\bigr] \,, 
\eeq
where 
$$Q(\mu) \:= -\frac{\rm id}{\mu} \+ 
\begin{pmatrix} 
0 & 0 \\ 
0 & \gamma(\mu) \\ \end{pmatrix}\,.
$$
\end{prop}
\begin{proof}
Define $W^*$ as the right-hand side of~\eqref{naR}, i.e., 
$$
W^*\:=\frac1{\mu-\lambda} \bigl[ R(\mu), R(\lambda) \bigr] + \bigl[Q(\mu),R(\lambda)\bigr] \,.
$$
More precisely, the entries of~$W^*$ have the expressions:
\begin{align}
& X^* \= \tfrac{w_0}{\mu -\lambda}  \Bigl(\tfrac{(\alpha(\lambda)+\Lambda(\alpha(\lambda))+1) (\Lambda^{-1}(\alpha(\mu))+\alpha(\mu)+1)}{(\lambda -v_0) (\mu -v_{-1})}  - \tfrac{(\Lambda^{-1}(\alpha(\lambda))+\alpha(\lambda)+1) (\alpha(\mu)+\Lambda(\alpha(\mu))+1)}{(\lambda -v_{-1}) (\mu -v_0 )}\Bigr),\\
& Y^* \= \tfrac{w_0}{\lambda -\mu } \Bigl(\tfrac{(\alpha(\lambda)+\Lambda(\alpha(\lambda))+1) (\Lambda^{-1}(\alpha(\mu)) (\lambda -\mu )+\alpha(\mu) (\lambda +\mu -2 v_{-1})+\lambda -v_{-1})}{(\lambda -v_0) (\mu -v_{-1})} +\tfrac{(2 \alpha(\lambda)+1) (\alpha(\mu)+\Lambda(\alpha(\mu))+1)}{v_0-\mu }\Bigr),\\
& Z^* \= \tfrac1{\lambda -\mu } \Bigl( \tfrac{(\Lambda^{-1}(\alpha(\lambda))+\alpha(\lambda)+1) (\Lambda^{-1}(\alpha(\mu))-\alpha(\mu))}{v_{-1}-\lambda } +\tfrac{(\Lambda^{-1}(\alpha(\lambda))-\alpha(\lambda)) (\Lambda^{-1}(\alpha(\mu))+\alpha(\mu)+1)}{\mu -v_{-1}} \Bigr).
\end{align}
We can then verify that $W^* \in \A\otimes {\rm sl}_2(\CC) [[ \lambda^{-1},\mu^{-1}]] \lambda^{-1} \mu^{-1}$, 
as well as that $W:=W^*$ satisfies the two equations~\eqref{W1}--\eqref{W2}.
The latter is done by a lengthy but straightforward calculation. 
The proposition is proved due to Lemma~\ref{observationDR}.
\end{proof}

If we define $\widetilde \Omega_{i,j}, \widetilde S_i$ by
\begin{align}
&\sum_{i,j\geq 0} 
\frac{\widetilde\Omega_{i,j}}{\lambda^{i+2} \mu^{j+2}} \= 
 \frac{\Tr \, \bigl(R(\lambda) R(\mu) \bigr)}{(\lambda-\mu)^2} 
\,-\, \frac{1}{(\lambda_1-\lambda_2)^2} \,,\\
&\Lambda \bigl(\gamma(\lambda)\bigr) \= \lambda^{-1}\+\sum_{i\geq 0} \widetilde S_i \, \lambda^{-i-2} \,,
\end{align}
then according to~\cite{DuY1}, 
$\widetilde \Omega_{i,j}, \widetilde S_i$ gives the canonical tau-structure for the Toda lattice, i.e.,
\[\widetilde \Omega_{i,j}\=\Omega_{i,j}, \quad \widetilde S_i\=S_i\,.\] 
These equalities together with Proposition~\ref{thelemma} 
lead to Proposition~\ref{matrixprop}; see~\cite{DuY1} for the detailed proof of Proposition~\ref{matrixprop}. 

Before ending this section, we will make two remarks.
The first remark is that all the entries of~$R(\lambda)$ 
can be expressed by the canonical tau-structure. Indeed, we have
\begin{align}
& \alpha(\lambda) \= \sum_{p\geq 0} \Omega_{p,0} \, \lambda^{-p-2}\,, 
\quad \beta(\lambda)=-w_0 \, \Lambda\bigl(\gamma(\lambda)\bigr) \,, \label{taualphabeta} \\
& \Lambda \bigl(\gamma(\lambda)\bigr) \= \lambda^{-1}\+\sum_{p\geq 0} S_p \, \lambda^{-p-2} \,. \label{taugamma}
\end{align}
The proof was in~\cite{DuY1}.
The second remark is 
 that existence of a tau-structure in general implies 
Lemma~\ref{commtoda}, and note that the proof in~\cite{DuY1} of the fact that 
$\widetilde \Omega_{i,j}, \widetilde S_i$ is a tau-structure 
 does not use the commutativity of the abstract Toda lattice hierarchy, 
 so, as a byproduct of the matrix resolvent method we get a new proof of Lemma~\ref{commtoda}
 together with a simple construction of the Toda lattice hierarchy. Similar idea was in~\cite{BDY3}.

\section{Pair of wave functions}\label{section3}
As in the Introduction, we start with the linear operator 
$L(n)= \Lambda+f(n) + g(n) \, \Lambda^{-1}$, where $f(n)$ and~$g(n)$ 
are two given arbitrary elements in~$V$.
We show in this section the existence of pairs of wave functions associated to $(f(n),g(n))$. 
Let us write
\begin{align}
& \psi_A(\lambda,n) \= e^{(\Lambda-1)^{-1}y(\lambda,n)}  \lambda^n \,, \quad y(\lambda,n)\:=\sum_{i\geq1} \frac{y_i(n)}{\lambda^i} \,, \\
& \psi_B(\lambda,n) \= e^{(\Lambda-1)^{-1}z(\lambda,n)}  e^{-s(n)} \, \lambda^{-n} \,, \quad z(\lambda,n)\:=\sum_{i\geq1} \frac{z_i(n)}{\lambda^i} \,.  
\end{align}
Then the spectral problems~$L(n) \bigl(\psi(\lambda,n)\bigr) = \lambda\psi(\lambda,n)$ for $\psi=\psi_A$ and for $\psi=\psi_B$ 
recast into the following equations:
\begin{align}
& \lambda \, e^{y(\lambda,n)} \+ f(n) \,-\, \lambda \+ g(n) \, \lambda^{-1} e^{-y(\lambda,n-1)} \= 0\,, \label{yexp} \\
& \lambda \, e^{-z(\lambda,n-1)} \+ f(n) \,-\, \lambda \+ g(n+1) \, \lambda^{-1} e^{z(\lambda,n)} \= 0\,, \label{zexp}
\end{align}
yielding recursions of the form (as equivalent conditions to~\eqref{yexp}--\eqref{zexp})
\begin{align}
& y_{k+1}(n) \= 
- \sum_{m_1,\dots,m_k\geq 0\atop \sum_{i=1}^{k} im_i=k+1} 
\frac{\prod_{i=1}^k y_i(n)^{m_i}}{\prod_{i=1}^k m_i!} 
-f(n)\delta_{k,0} \nn\\
&\qquad\qquad\qquad -g(n) \sum_{m_1,\dots,m_{k-1}\geq 0\atop \sum_{i=1}^{k-1} im_i=k-1} 
\frac{\prod_{i=1}^{k-1} (-1)^{m_i} y_i(n-1)^{m_i}}{\prod_{i=1}^{k-1} m_i!} \,, \label{ykn}\\
& z_{k+1}(n) \= 
\sum_{m_1,\dots,m_k\geq 0\atop \sum_{i=1}^{k} im_i=k+1} 
\frac{\prod_{i=1}^k (-1)^{m_i} z_i(n)^{m_i}}{\prod_{i=1}^k m_i!} 
\+ f(n+1)\delta_{k,0} \nn\\
&\qquad\qquad\qquad \+ g(n+2) \sum_{m_1,\dots,m_{k-1}\geq 0\atop \sum_{i=1}^{k-1} im_i=k-1} 
\frac{\prod_{i=1}^{k-1} z_i(n+1)^{m_i}}{\prod_{i=1}^{k-1} m_i!} \,, \label{zkn}
\end{align}
where $k\geq 0$. From these recursions, it easily follows that $y_k,z_k\in V$, $k\geq 0$. This proves the 
existence of wave functions of type A and of type B meeting the definitions in Section~\ref{Introwave}. 
Clearly, $\psi_A$ and $\psi_B$ are unique up to multiplying by 
arbitrary series  $G(\lambda)$ and $E(\lambda)$ of~$\lambda^{-1}$ with constant coefficient of the form 
$G(\lambda)\in 1+ \CC[[\lambda^{-1}]]\lambda^{-1}$ and 
$E(\lambda) \in 1+ \CC[[\lambda^{-1}]]\lambda^{-1}$.
Since $\psi_A(\lambda,n)=\bigl(1+{\rm O}(\lambda^{-1})\bigr) \,\lambda^n$ 
and since $\psi_B(\lambda,n)=\bigl(1+{\rm O}(\lambda^{-1})\bigr) \, e^{-s(n)} \lambda^{-n}$, 
we find that the $d(\lambda,n)$ defined in~\eqref{definitiondtimeind} must have the form
$$
d(\lambda,n)\= \lambda \, e^{-s(n-1)} \, e^{\sum_{k\geq 1} d_k(n) \, \lambda^{-k}} \,. 
$$
Then by using the definitions of wave functions and of~$s(n)$ one easily derives that 
\beq \label{alreadyprovedsd}
e^{s(n)} \, d(\lambda,n+1)\= e^{s(n-1)} \, d(\lambda,n) \,. 
\eeq
It follows that all $d_k(n)$, $k\geq 1$ are constants. Therefore, for any fixed choice of $\psi_A$, 
we can suitably choose the factor $E(\lambda)$ for~$\psi_B$ such that $\psi_A,\psi_B$ form a pair. 
This proves the existence of pair of wave functions associated to $f(n),g(n)$.

We proceed with the time-dependence. Let $(v(n,\bt),w(n,\bt))$ be the unique solution 
in $V[[\bt]]^2$ to the Toda lattice hierarchy
satisfying the initial condition $v(n,\bdzero)=f(n)$, $w(n,\bdzero)=g(n)$. 
Let $L(n,\bt) := \Lambda+v(n,\bt) + w(n,\bt) \, \Lambda^{-1}$. 
Define $\sigma(n,\bt)$ as the unique up to a constant function satisfying the following equations: 
\begin{align}
& w(n,\bt) \= e^{\sigma(n-1,\bt)-\sigma(n,\bt)} \,, \\ 
& \frac{\p \sigma(n,\bt)}{\p t_p} \=  
- S_p (n,\bt)\,,\quad p\geq 0\,.
\label{Toda-q}
\end{align}
An element $\psi_A(n,\bt,\lambda)= (1+{\rm O}(\lambda^{-1})) \, \lambda^n \, e^{\sum_{k\geq 0} t_k\lambda^{k+1}}$ in  
$\widetilde V\bllm\bt, \lambda^{-1} \brrm \, \lambda^n e^{\sum_{k\geq 0} t_k\lambda^{k+1}}$ 
is called a wave function of type A associated to $(v(n,\bt),w(n,\bt))$ if 
\begin{align}
& L(n,\bt) \, \bigl(\psi_A(\lambda,n,\bt)\bigr) \= \lambda \, \psi_A(\lambda,n,\bt)  \,, \quad 
\frac{\p\psi_A}{\p t_k} \= \bigl(L^{k+1}\bigr)_+ \bigl(\psi_A\bigr) \,. \label{wave112}
\end{align}
An element $\psi_B(n,\bt,\lambda) = (1+{\rm O}(\lambda^{-1})) \lambda^{-n} e^{-\sum_{k\geq 0} t_k\lambda^{k+1}}$ in 
$\widetilde V\bllm\bt,\lambda^{-1} \brrm  e^{-\sigma(n,\bt)} \lambda^{-n} e^{-\sum_{k\geq 0} t_k\lambda^{k+1}}$
is called a wave function of type B associated to $(v(n,\bt),w(n,\bt))$ if 
\begin{align}
& L(n,\bt) \, \bigl(\psi_B(\lambda,n,\bt)\bigr) \= \lambda \, \psi_B(\lambda,n,\bt)  \,, \quad 
\frac{\p\psi_B}{\p t_k} \=  - \bigl(L^{k+1}\bigr)_- \bigl(\psi_B\bigr) \,. \label{wave212}
\end{align}
The existence of wave functions $\psi_A$ and $\psi_B$ of type~A and of type~B associated to $(v(n,\bt),w(n,\bt))$
is a standard result in the theory of integrable systems (cf.~\cite{UT, CDZ,Carlet,DYZ2}); therefore we omit its details.
Denote 
\beq\label{wavepairdd}
d(\lambda,n,\bt)\:= \psi_A(\lambda,n,\bt) \, \psi_B(\lambda,n-1,\bt) \,-\, \psi_B(\lambda,n,\bt) \, \psi_A(\lambda,n-1,\bt) \,,
\eeq
and introduce  
\beq\label{definitionm}
m(\mu, \lambda, n,\bt) \:= \frac{R(\mu,n,\bt)}{\mu-\lambda} \+ Q(\mu,n,\bt)\,,
\eeq
where
$Q(\mu,n,\bt) := -\frac{\rm id}{\mu} \+  \begin{pmatrix} 0 & 0 \\ 0 & \gamma(\mu,n,\bt) \\ \end{pmatrix}$.
We know from e.g.~\cite{DuY1} that the wave function~$\psi_A(\lambda,n,\bt)$ satisfies 
\begin{align}
& \nabla(\mu) \, \begin{pmatrix} \psi_A(\lambda,n,\bt) \\ \psi_A(\lambda,n-1,\bt) \end{pmatrix} 
\= 
m(\mu,\lambda,n,\bt) \, \begin{pmatrix}  \psi_A(\lambda,n,\bt) 
\\ \psi_A(\lambda,n-1,\bt) \end{pmatrix} \,.\label{nablapsiA}
\end{align}
Similarly, the wave function~$\psi_B(\lambda,n,\bt)$ satisfies 
\begin{align}
& \nabla(\mu) \, \begin{pmatrix} \psi_B(\lambda,n,\bt) \\ \psi_B(\lambda,n-1,\bt) \end{pmatrix} 
\= 
\biggl(m(\mu,\lambda,n,\bt) \, -  \,  \frac{\lambda}{\mu(\mu-\lambda)} I \biggr)\, \begin{pmatrix}  \psi_B(\lambda,n,\bt) 
\\ \psi_B(\lambda,n-1,\bt) \end{pmatrix} \,.\label{nablapsiB}
\end{align}
Here, $I$ denotes the $2\times 2$ identity matrix.

\begin{lemma}\label{nablad} The following formula holds true:
\beq
\nabla (\mu) \, \bigl(d(\lambda,n,\bt)\bigr)
 \=  \biggl( - \frac1{\mu} + \gamma(\mu,n,\bt)  \biggr) \, d(\lambda,n,\bt)\,.
\eeq
\end{lemma}
\begin{proof}
Recalling the definition~\eqref{wavepairdd} for~$d$ and 
using~\eqref{nablapsiA}--\eqref{nablapsiB} we find
\begin{align}
\nabla (\mu) \, \bigl( d(\lambda,n,\bt) \bigr)
& \=  \biggl( \tr \bigl(m(\mu,\lambda,n,\bt)\bigr) - \frac{\lambda}{\mu(\mu-\lambda)} \biggr) \, d(\lambda,n,\bt)\,.
\end{align}
The lemma is then proved via a straightforward computation.
\end{proof}

\begin{defi}
We say $\psi_A,\psi_B$ form {\it a pair} if $ e^{\sigma(n-1,\bt)} d(\lambda,n,\bt)=\lambda$.
\end{defi}

The next lemma shows the existence of a pair.

\begin{lemma}\label{existencepair}
There exist a pair of wave functions $\psi_A,\psi_B$ associated to $(v(n,\bt),w(n,\bt))$. Moreover, 
the freedom of the pair is characterized by a factor~$G(\lambda)$ via
\begin{align}
& \psi_A(\lambda,n,\bt) \; \mapsto \; G(\lambda) \,  \psi_A(\lambda,n,\bt) \,, \quad 
\psi_B(\lambda,n,\bt) \; \mapsto \; \frac1{G(\lambda)} \,  \psi_B(\lambda,n,\bt) \,, \label{psi12g}  \\
& G(\lambda) \= \sum_{j\geq 0} G_j \lambda^{-j}\,, \quad G_0\=1  \label{glambda}
\end{align}
with $G_j$, $j\geq 1$ being arbitrary constants.
\end{lemma}
\begin{proof} 
Firstly, the freedom of a wave function~$\psi_A$ 
associated to $(v,w)$ is 
characterized by the multiplication by a factor $G(\lambda)$ of the form~\eqref{glambda}.
Fix an arbitrary choice of~$\psi_A$. For $\psi_B$ being 
a wave function of type~B associated to $(v,w)$, 
from~\eqref{wavepairdd} and the definitions of wave functions we know $e^{\sigma(n-1,\bt)}d(\lambda,n,\bt)$ must have 
the form
\beq\label{deduce0}
e^{\sigma(n-1,\bt)}d(\lambda,n,\bt) \= \lambda \, e^{\sum_{k\geq 1} d_k(n,\bt) \, \lambda^{-k}}
\eeq
for some $d_k(n,\bt)$, $k\geq 1$.
By using \eqref{wave112}, \eqref{wave212}, \eqref{wavepairdd} we find 
\begin{align}
& d(\lambda,n+1,\bt) \= w(n,\bt) \, d(\lambda,n,\bt) \= e^{\sigma(n-1,\bt)-\sigma(n,\bt)} \, d(\lambda,n,\bt)  \,,\nn
\end{align}
i.e.,
\beq \label{deduce1}
e^{\sigma(n,\bt)} d(\lambda,n+1,\bt)\= e^{\sigma(n-1,\bt)} \, d(\lambda,n,\bt)  \,, 
\eeq
Using Lemma~\ref{nablad} and~\eqref{Toda-q} we have 
\begin{align}
& \quad\quad \nabla(\mu) \Bigl(e^{\sigma(n-1,\bt)}d(\lambda,n,\bt)\Bigr) \nn\\
&  \= e^{\sigma(n-1,\bt)} \nabla(\mu) \bigl(\sigma(n-1,\bt)\bigr) d(\lambda,n,\bt) \+ 
e^{\sigma(n-1,\bt)} \nabla(\mu) \bigl(d(\lambda,n,\bt)\bigr) \nn\\
&  \= - e^{\sigma(n-1,\bt)}\sum_{p\geq 0} \frac{S_p(n-1,\bt)}{\mu^{p+2}} d(\lambda,n,\bt) \+ 
e^{\sigma(n-1,\bt)} d(\lambda,n,\bt) \biggl( - \frac1{\mu} \+ \gamma(\mu,n,\bt)  \biggr) \= 0 \,.\nn
\end{align}
So we have
\beq \label{deduce2}
\frac{\p (e^{\sigma(n-1,\bt)}d(\lambda,n,\bt))}{\p t_p} \= 0 \,, \quad \forall\, p\geq 0\,. 
\eeq
We deduce from~\eqref{deduce0}, \eqref{deduce1}, \eqref{deduce2} that $d_k(n,\bt)$, $k\geq 1$ are all constants. 
Therefore, there exists a unique choice of~$\psi_B$ such that 
$\psi_A,\psi_B$ form a pair. The lemma is proved. 
\end{proof} 

\section{The $k$-point generating series} \label{section4}
Let $(v,w)=(v(n,\bt),w(n,\bt))\in V[[\bt]]^2$ be the unique solution to the Toda lattice hierarchy 
with the initial value $(v(n,\bdzero),w(n,\bdzero))=(f(n),g(n))$, and $(\psi_A,\psi_B)$ a pair of wave functions associated to~$(v,w)$.  
Define 
\beq\label{Psidefinition}
\Psi_{\rm pair}(\lambda,n,\bt) \= \begin{pmatrix} \psi_A (\lambda,n,\bt) & \psi_B(\lambda,n,\bt) 
\\ \psi_A(\lambda,n-1,\bt) & \psi_B(\lambda,n-1,\bt) \end{pmatrix}\,. 
\eeq
\begin{prop} \label{propRP}
The following identity holds true: 
\beq\label{RP}
R(\lambda, n, \bt) \; \equiv \;   \Psi_{\rm pair}(\lambda,n,\bt) \, \begin{pmatrix} 1 & 0 \\ 0 & 0 \end{pmatrix} 
\, \Psi_{\rm pair}^{-1}(\lambda,n,\bt)  \,.  
\eeq
\end{prop}
\begin{proof} Define   
\[M\=M(\lambda,n,\bt)\:= 
\Psi_{\rm pair}(\lambda,n,\bt) \, \begin{pmatrix} 1 & 0 \\ 0 & 0 \end{pmatrix} \, \Psi_{\rm pair}^{-1}(\lambda,n,\bt)\,.\]
It is easy to verify that $M$ satisfies 
\[ 
\bigl[\mathcal{L},M\bigr] \,  \bigl(\Psi_{\rm pair}\bigr) \= 0 \,, \qquad \det \, M \= 0 \,.
\]
The entries of~$M$ in terms of the pair of wave functions read
\begin{align}
M=\frac1{d(\lambda,n,\bt)}
 \begin{pmatrix}  
 \psi_A(\lambda,n,\bt) \, \psi_B(\lambda,n-1,\bt) & -\psi_A(\lambda,n,\bt) \, \psi_B(\lambda,n,\bt) \\  
 \psi_A(\lambda,n-1,\bt) \, \psi_B(\lambda,n-1,\bt) & -\psi_A(\lambda,n-1,\bt) \, \psi_B(\lambda,n,\bt)
  \end{pmatrix}  \,, 
\end{align}
where we recall that $d(\lambda,n,\bt) = \psi_A(\lambda,n,\bt) \, \psi_B(\lambda,n-1,\bt) - \psi_B(\lambda,n,\bt) \, \psi_A(\lambda,n-1,\bt)$, which coincides with the determinant of~$\Psi(\lambda,n,\bt)$. 
It follows from $\psi_A(\lambda,n,\bt)= (1+{\rm O}(\lambda^{-1})) \, \lambda^n \, e^{\sum_{k\geq 0} t_k\lambda^{k+1}}$ and 
$\psi_B(\lambda,n,\bt) =  (1+{\rm O}(\lambda^{-1}))\, e^{-\sigma(n,\bt)} \lambda^{-n} \, e^{-\sum_{k\geq 0} t_k\lambda^{k+1}}$ that 
\beq
M(\lambda) \,-\, \begin{pmatrix} 1 & 0 \\ 0 & 0 \end{pmatrix} 
\; \in \; {\rm Mat}\left(2,\widetilde V[[\bt,\lambda^{-1}]]\lambda^{-1}\right) \,.
\eeq
The proposition then follows from the uniqueness theorem proven in Section~\ref{section2}.
\end{proof}

Define
\beq\label{defineD}
D(\lambda,\mu,n,\bt) \:= \frac{\psi_A(\lambda,n,\bt) \, \psi_B(\mu,n-1,\bt) \,-\, \psi_A(\lambda,n-1,\bt) \,\psi_B(\mu,n,\bt)}{\lambda-\mu} \,.
\eeq

\begin{theorem}\label{main1time} 
Fix $k\geq 2$ being an integer. 
The generating series of $k$-point correlation functions of the solution $(v(n,\bt),w(n,\bt))$ has the following expression:
\begin{align}
& \sum_{i_1,\dots,i_k\geq 0} 
\frac{\Omega_{i_1,\dots,i_k}(n,\bt)}{\lambda_1^{i_1+2} \cdots \lambda_k^{i_k+2}} \nn\\
& \qquad \qquad  \=  
(-1)^{k-1} \frac{e^{k \sigma(n-1,\bt)}}{\prod_{j=1}^k \lambda_{j}}  
\sum_{\pi \in \mathcal{S}_k/C_k} \prod_{j=1}^k D(\lambda_{\pi(j)},\lambda_{\pi(j+1)},n,\bt) 
 \,-\, \frac{\delta_{k,2}}{(\lambda_1-\lambda_2)^2}  \,.  \label{mainidentity}
\end{align}
\end{theorem}
\begin{proof}
It follows from~\eqref{RP} that 
\beq\label{RP1}
R(\lambda, n, \bt) \=    
\frac{r_1(\lambda,n,\bt)^T r_2(\lambda,n,\bt)}{d(\lambda,n,\bt)}   \,,
\eeq
where $r_1(\lambda,n,\bt):=(\psi_A(\lambda,n,\bt), \psi_A(\lambda,n-1,\bt))$, 
$r_2(\lambda,n,\bt):=(\psi_B(\lambda,n-1,\bt),-\psi_B(\lambda,n,\bt))$. 
Substituting this expression into the identity 
\begin{align}
& \sum_{i_1,i_2\geq 0} 
\frac{\Omega_{i_1,i_2}(n,\bt)}{\lambda_1^{i_1+2} \lambda_2^{i_2+2}} \= 
\frac{\Tr \, \bigl(R_1(\lambda_1,n,\bt) \, R_2(\lambda_2,n,\bt)\bigr)}
{(\lambda_1-\lambda_2)^2} 
\,-\, \frac{1}{(\lambda_1-\lambda_2)^2} \,, 
\end{align}
we obtain
\begin{align}
\sum_{i_1,i_2\geq 0} 
\frac{\Omega_{i_1,i_2}(n,\bt)}{\lambda_1^{i_1+2} \lambda_2^{i_2+2}} 
& \= 
\frac{\Tr \, \bigl(r_1(\lambda_1,n,\bt)^T r_2(\lambda_1,n,\bt) \, r_1(\lambda_2,n,\bt)^T r_2(\lambda_2,n,\bt)\bigr)}
{d(\lambda_1,n,\bt) d(\lambda_2,n,\bt) (\lambda_1-\lambda_2)^2} 
\,-\, \frac{1}{(\lambda_1-\lambda_2)^2} \nn \\
& \= 
\frac{\bigl(r_2(\lambda_2,n,\bt) \, r_1(\lambda_1,n,\bt)^T \bigr) \, \bigl(r_2(\lambda_1,n,\bt) \, r_1(\lambda_2,n,\bt)^T \bigr)}
{ d(\lambda_1,n,\bt) d(\lambda_2,n,\bt)  (\lambda_1-\lambda_2)^2} 
\,-\, \frac{1}{ (\lambda_1-\lambda_2)^2} \nn \\
&\= - \frac{D(\lambda_1,\lambda_2,n,\bt) \, 
D(\lambda_2,\lambda_1,n,\bt) }{ \lambda_1 \lambda_2 \, e^{-2\sigma(n-1,\bt)} } \,-\, \frac{1}{(\lambda_1-\lambda_2)^2} \,,
\end{align}
where we used the definition~\eqref{defineD} and 
\begin{align}
\frac{\psi_A(\lambda,n,\bt) \, \psi_B(\mu,n-1,\bt) \,-\, \psi_A(\lambda,n-1,\bt) \,\psi_B(\mu,n,\bt)}{\lambda-\mu}  
\=  \frac{r_2(\mu,n,\bt) \, r_1(\lambda,n,\bt)^T}{\lambda-\mu} \,. \nn
\end{align}
This prove the $k=2$ case of~\eqref{mainidentity}. For $k\geq 3$, the proof is similar. Indeed,
\begin{align}
& \sum_{i_1,\dots,i_k\geq 0} 
\frac{\Omega_{i_1,\dots,i_k}(n,\bt)}{\lambda_1^{i_1+1} \cdots \lambda_k^{i_k+1}} \nn\\
&  \qquad \= - \sum_{\pi\in \cS_k/C_k} 
\frac{\Tr \, \Bigl(\prod_{j=1}^k  r_1\bigl(\lambda_{\pi(j)},n,\bt\bigr)^T r_2\bigl(\lambda_{\pi(j)},n,\bt\bigr) \Bigr)}
{e^{-k \, \sigma(n-1,\bt)}  \, \prod_{j=1}^k \bigl(\lambda_{\pi{(j)}}-\lambda_{\pi{(j+1)}}\bigr)} \nn\\
&  \qquad \= 
- \sum_{\pi \in \cS_k/C_k} 
\frac{r_2\bigl(\lambda_{\pi(k)},n,\bt\bigr) \, r_1\bigl(\lambda_{\pi(1)},n,\bt\bigr)^T \, \cdots \,  
r_2\bigl(\lambda_{\pi(k-1)},n,\bt\bigr) \, r_1\bigl(\lambda_{\pi(k)},n,\bt\bigr)^T}{e^{-k \, \sigma (n-1,\bt)}  \, 
\prod_{j=1}^k \bigl(\lambda_{\pi{(j)}}-\lambda_{\pi{(j+1)}}\bigr)} \nn\\
&  \qquad \=  
- \frac{(-1)^k}{e^{-k \, \sigma(n-1,\bt)}} 
\sum_{\pi\in \cS_k/C_k} \prod_{j=1}^k D\bigl(\lambda_{\pi(j)},\lambda_{\pi(j+1)},n,\bt\bigr)\,.
\end{align}
This proves the $k\geq 3$ case of~\eqref{mainidentity}. The theorem is proved. 
\end{proof}
\begin{remark}
In~\eqref{mainidentity} or~\eqref{mainidentityfg}, the freedom~\eqref{psi12g} affects the $D(\lambda, \mu)$ 
through multiplying it by
a factor of the form $\frac{G(\lambda)}{G(\mu)}$, but the product
$\prod_{j=1}^k D(\lambda_{\pi(j)},\lambda_{\pi(j+1)})$ 
remains unchanged. 
\end{remark}

In Appendix~A, the abstract form of~\eqref{mainidentity} is obtained, 
where a pair of abstract pre-wave functions are introduced.

\noindent {\it Proof} of Theorem~\ref{main1}.  Taking ${\bf t}=\bdzero$ on the 
both sides of~\eqref{mainidentity} gives~\eqref{mainidentityfg}. \epf

Write 
\beq
\psi_A(\lambda,n,\bt) \= \phi_A(\lambda,n,\bt) \,  \lambda^n \,, \qquad \psi_B(\lambda,n,\bt) \= \phi_B(\lambda,n,\bt) \, e^{-\sigma(n,\bt)} \, \lambda^{-n} \,.
\eeq
Theorem~\ref{main1} can then be alternatively written in terms of $\phi_A,\phi_B$ by the following corollary. 
\begin{cor}
The following formula holds true for $k\geq2$:
\begin{align}
& \sum_{i_1,\dots,i_k\geq 0} 
\frac{\Omega_{i_1,\dots,i_k}(n,\bt)}{\lambda_1^{i_1+2} \cdots \lambda_k^{i_k+2}}  \=  
(-1)^{k-1}   
\sum_{\pi \in \mathcal{S}_k/C_k} \prod_{j=1}^k B(\lambda_{\pi(j)},\lambda_{\pi(j+1)},n,\bt) 
\,-\, \frac{\delta_{k,2}}{(\lambda_1-\lambda_2)^2}  \,, 
\end{align}
where $B(\lambda,\mu,n,\bt)$ is defined by
\beq
B(\lambda,\mu,n,\bt) \:= 
\frac{\phi_A(\lambda,n,\bt) \, \phi_B(\mu,n-1,\bt) \,-\, w(n,\bt)\, \phi_A(\lambda,n-1,\bt) \,\phi_B(\mu,n,\bt)}{\lambda-\mu} \,.
\eeq
In particular,  let 
$\phi_A(\lambda,n):=e^{(\Lambda-1)^{-1}(y(\lambda,n))}$, 
$\phi_B(\lambda,n):=e^{(\Lambda-1)^{-1}(z(\lambda,n))}  e^{-s(n)}$ (cf.~\eqref{yexp}--\eqref{zexp}), and let 
$B(\lambda,\mu,n):= \frac{\phi_A(\lambda,n) \, \phi_B(\mu,n-1) \,-\, g(n)\, \phi_A(\lambda,n-1) \,\phi_B(\mu,n)}{\lambda-\mu}$,
then we have
\begin{align}
& \sum_{i_1,\dots,i_k\geq 0} 
\frac{\Omega_{i_1,\dots,i_k}(n,\bdzero)}{\lambda_1^{i_1+2} \cdots \lambda_k^{i_k+2}}  \=  
(-1)^{k-1}   
\sum_{\pi \in \mathcal{S}_k/C_k} \prod_{j=1}^k B(\lambda_{\pi(j)},\lambda_{\pi(j+1)},n) 
\,-\, \frac{\delta_{k,2}}{(\lambda_1-\lambda_2)^2}  \,.
\end{align}
\end{cor}

For some particular examples related to matrix models, 
it turns out that the suitable chosen~$D$ coincides, possibly up to simple factors, 
with certain kernel of the matrix model. However,   
the~$D$ is not unique. We now introduce a formal series  
$K(\lambda,\mu)$ such that the generating series of multi-point correlation functions still has an explicit expression, but 
this time $K$ is {\it local} and is therefore unique for the given solution. The series~$K$ is defined by
\beq
K(\lambda,\mu) \:= 
 \frac{(1\+\alpha(\lambda))(1\+\alpha(\mu)) -  w_0 \, 
\gamma(\lambda) \, \Lambda \bigl(\gamma(\mu)\bigr)}{\lambda-\mu}\,,
\eeq
where $1+\alpha(\lambda)$ is the (1,1)-entry of the basic matrix resolvent~$R(\lambda)$, and 
$\gamma(\lambda)$ is the (2,1)-entry.
The next theorem expresses the left-hand side of~\eqref{mainidentity} in terms of~$K$.

\begin{theorem}\label{main2} For any $k\geq 2$,
the following formula holds true:
\begin{align}
&   \sum_{i_1,\dots,i_k\geq 0} 
\frac{\Omega_{i_1,\dots,i_k}}{\lambda_1^{i_1+2} \cdots \lambda_k^{i_k+2}}  \=  
(-1)^{k-1} \frac{\sum_{\pi\in \cS_k/C_k} \prod_{j=1}^k K\bigl(\lambda_{\pi(j)},\lambda_{\pi(j+1)}\bigr)}{ \prod_{i=1}^k \bigl(1\+\alpha(\lambda_i)\bigr)} \,-\, \frac{\delta_{k,2}}{(\lambda_1-\lambda_2)^2}   \,.
\end{align}
\end{theorem}

\begin{proof} The identity~\eqref{RP} gives  
\begin{align}
& \psi_B(\lambda,n-1,\bt) \= 
\frac{(1\+\alpha(\lambda,n,\bt)) \, d(\lambda,n,\bt)}{\psi_A(\lambda,n,\bt)} \,, \nn\\
& \psi_B(\lambda,n,\bt) \= 
- \, \frac{\beta(\lambda,n,\bt) \, d(\lambda,n,\bt)}{\psi_A(\lambda,n,\bt)} \= w_n\,\frac{\gamma(\lambda,n+1,\bt) \, d(\lambda,n,\bt)} {\psi_A(\lambda,n,\bt)} \,, \nn\\
& \psi_A(\lambda,n-1,\bt)  \= \psi_A(\lambda,n,\bt) \, 
\frac{\gamma(\lambda,n,\bt)}{1\+\alpha(\lambda,n,\bt)} \nn \,.
\end{align}
Substituting these expressions into~\eqref{defineD} we obtain
\beq
D(\lambda,\mu,n,\bt) \= d(\mu,n,\bt) \frac{\psi_A(\lambda,n,\bt)}{\psi_A(\mu,n,\bt)} e(\lambda,\mu,n,\bt),
\eeq
where 
\beq
e(\lambda,\mu,n,\bt) \:= \frac{(1\+\alpha(\lambda,n,\bt))(1\+\alpha(\mu,n,\bt)) -  w_n(\bt) \, 
\gamma(\lambda,n,\bt) \, \gamma(\mu,n+1,\bt)
}{(\lambda-\mu) \, (1\+\alpha(\lambda,n,\bt))} \,. 
\eeq
Combining with the definition of~$K(\lambda,\mu,n,\bt)$ and Theorem~\ref{main1} we find
\begin{align}
&   \sum_{i_1,\dots,i_k\geq 0} 
\frac{\Omega_{i_1,\dots,i_k}(n,\bt)}{\lambda_1^{i_1+2} \cdots \lambda_k^{i_k+2}}  \=  
(-1)^{k-1} \sum_{\pi\in \cS_k/C_k} \prod_{j=1}^k K \bigl(\lambda_{\pi(j)},\lambda_{\pi(j+1)},n,\bt\bigr) 
 \,-\, \frac{ \delta_{k,2}}{(\lambda_1-\lambda_2)^2} \,.
\end{align}
The theorem is proved. 
\end{proof}

It seems to be an interesting question to study the geometric and algebraic 
meaning of the kernel~$K$ (as well as~$D$). Below we give without proof some of their properties.

\begin{prop} \label{KandD}
The functions $K$ and $D$ are related by 
\begin{align}
& K(\lambda,\mu,n,\bt) \=\frac{e^{\sigma(n-1,\bt)} }{\mu} \, \bigl(1+\alpha(\lambda,n,\bt)\bigr) \, 
\frac{\psi_A(\mu,n,\bt)}{\psi_A(\lambda,n,\bt)} \, D(\lambda,\mu,n,\bt) \nn\\
& \quad \quad \quad \quad \quad 
\=\frac{e^{2\sigma(n-1,\bt)}}{\lambda\,\mu} \, \psi_A(\mu,n,\bt) \, \psi_B(\lambda,n-1,\bt)\, D(\lambda,\mu,n,\bt) \nn\\
& \quad \quad \quad \quad \quad 
\=\frac{e^{\sigma(n-1,\bt)}}{\lambda} \, \bigl(1+\alpha(\mu,n,\bt)\bigr)\,\frac{\psi_B(\lambda,n-1,\bt) }{\psi_B(\mu,n-1,\bt)} \, D(\lambda,\mu,n,\bt) \,. \nn
\end{align}
\end{prop}

We observe that the following three formal series
$$
K(\lambda,\mu)-\frac{1+\alpha(\lambda)}{\lambda-\mu}\,, \quad K(\lambda,\mu)-\frac{1+\alpha(\mu)}{\lambda-\mu}\,, 
\quad K(\lambda,\mu)-\frac{2+\alpha(\lambda)+\alpha(\mu)}{2 \, (\lambda-\mu)}
$$ 
all belong to $\A[[\lambda^{-1},\mu^{-1}]]$. Therefore, the following 
three formal series
$$
K(\lambda,\mu,n,\bt)-\frac{1+\alpha(\lambda,n,\bt)}{\lambda-\mu}\,, 
\quad K(\lambda,\mu,n,\bt)-\frac{1+\alpha(\mu,n,\bt)}{\lambda-\mu}\,,
$$
$$
K(\lambda,\mu,n,\bt)-\frac{2+\alpha(\lambda,n,\bt)+\alpha(\mu,n,\bt)}{2(\lambda-\mu)}
$$ 
all belong to $V[[\bt]][[\lambda^{-1},\mu^{-1}]]$.
It follows from this observation and Proposition~\ref{KandD} that 
\beq\label{DD}
\frac{e^{s(n-1)} }\mu \, D(\lambda,\mu,n,\bdzero) \, \biggl(\frac{\mu}{\lambda}\biggr)^n \,-\, \frac{1}{\lambda-\mu}  
\; \in\;  \widetilde  V[[\lambda^{-1},\mu^{-1}]] \,.  
\eeq

\begin{remark}
We could loosen both the conditions for wave functions and the pair-condition.
Let us say $\psi_A$ and~$\psi_B$ are pre-wave functions of type~A and of type~B, respectively,  
if they satisfy the first equations of~\eqref{wave112} and~\eqref{wave212}. Define $d_{\rm pre}(\lambda,n,\bt)$ 
and~$D_{\rm pre}(\lambda,\mu,n,\bt)$ by~\eqref{wavepairddpre} and~\eqref{defineDpre}. 
Then the following formula holds true:
\begin{align}
& \sum_{i_1,\dots,i_k\geq 0} 
\frac{\Omega_{i_1,\dots,i_k}(n,\bt)}{\lambda_1^{i_1+2} \cdots \lambda_k^{i_k+2}}  \nn\\
&  \quad \qquad \=  
 \frac{(-1)^{k-1}}{\prod_{j=1}^k d_{\rm pre}(\lambda_j,n,\bt)}  
\sum_{\pi \in \cS_k/C_k} \prod_{j=1}^k D_{\rm pre}(\lambda_{\pi(j)},\lambda_{\pi(j+1)},n,\bt) 
 \,-\, \frac{\delta_{k,2}}{(\lambda_1-\lambda_2)^2}  \,.  \label{mainidentitypretime}
\end{align}
Now $\psi_A$ and $\psi_B$ are determined by $(v(n,\bt),w(n,\bt) )$ up to 
\[ \psi_A (\lambda,n,\bt) \; \mapsto\; G(\lambda,\bt) \, \psi_A(\lambda,n,\bt) \,, 
\quad  \psi_B (\lambda,n,\bt) \; \mapsto\; E(\lambda,\bt) \, \psi_B (\lambda,n,\bt) \,,\]
where 
$G(\lambda,\bt)=1+\sum_{k\geq 1} G_k(\bt) \lambda^{-k}$, $E(\lambda,\bt)=1+\sum_{k\geq 1} E_k(\bt)\lambda^{-k}$ 
with $G_k(\bt), E_k(\bt) \in \CC[[\bt]]$, $k\geq 1$.
This freedom affects $D_{\rm pre}(\lambda,\mu,n,\bt)$ and~$d_{\rm pre}(\lambda,n,\bt)$ into 
\[
D_{\rm pre}(\lambda,\mu,n,\bt)  \mapsto  G(\lambda,\bt) \, E(\mu,\bt) \, D_{\rm pre}(\lambda,\mu,\bt)  \,, 
\quad d_{\rm pre}(\lambda,n,\bt) \mapsto  G(\lambda,\bt) \, E(\lambda,\bt) \, d_{\rm pre}(\lambda,\bt) \,.
\]
Therefore, it gives rise to each summand of~\eqref{mainidentitypretime} the factor 
\[\frac{\prod_{j=1}^k G(\lambda_{\pi(j)},\bt) E(\lambda_{\pi(j+1)},\bt)}{\prod_{j=1}^k G(\lambda_j,\bt) E(\lambda_j,\bt)}\,,\]
which is equal to~1. Hence the right-hand side of~\eqref{mainidentitypretime} remains unchanged. 
\end{remark}

\section{Applications} \label{section5}
Partition functions in some matrix models and enumerative models are 
particular tau-functions for the Toda lattice hierarchy. Theorem~\ref{main1} can then be 
used for computing their logarithmic derivatives. In this section we do two   
explicit computations.

\subsection{Application I. Enumeration of ribbon graphs.}
The initial data of the GUE solution to the Toda lattice hierarchy is given by 
 $f(n)=0$ and $g(n)=n$; see for example~\cite{DuY1} for the proof. 
 For this case, we can take $V=\QQ[n]$ and $\widetilde V=V$.
Substituting the initial data in~\eqref{inir} we find
\beq
s(n) \= -  \bigl(1-\Lambda^{-1}\bigr)^{-1} \log g(n)  \= 
- \bigl(1-\Lambda^{-1}\bigr)^{-1} \log n\=-\log\,\Gamma(n+1) \+C\,,
\eeq
where $C$ is a constant. Below we fix this constant as~$0$.
\begin{prop}\label{GUEpsi}
The $\psi_A,\psi_B$ defined by 
\begin{align}
&\psi_A(\lambda,n)\= \sum_{j\geq 0} (-1)^j \frac{(n-2j+1)_{2j}}{2^j\, j!\, \lambda^{2j}} \lambda^n\,, \label{pair11gue}\\
&\psi_B(\lambda,n)\= \Gamma(n+1) \sum_{j\geq 0} \frac{(n+1)_{2j}}{2^j\,j!\, \lambda^{2j}} \lambda^{-n} \label{pair12gue}
\end{align}
form a particular pair of wave functions associated to $f(n),g(n)$. Here and below $(a)_{i}$ denotes the 
increasing Pochhammer symbol defined by $(a)_{i}=a(a+1)\cdots(a+i-1)$.
\end{prop}
\begin{proof}
It is straightforward to verify that both $\psi_A$ and $\psi_B$ satisfy the equation
\beq
\psi(\lambda,n+1) \+ n\, \psi(\lambda,n-1) \= \lambda \, \psi(\lambda,n) \,. 
\eeq
Moreover, from the definitions~\eqref{pair11gue}--\eqref{pair12gue} we see that 
\[\psi_A \in \widetilde V\bllb \lambda^{-1}\brrb \, \lambda^n\,, \quad \psi_B \in \widetilde V\bllb \lambda^{-1}\brrb  e^{-s(n)} \lambda^{-n}\,.\]
We are left to show that 
\beq\label{toshowab}
\Gamma(n)^{-1} \, \Bigl(\psi_A(\lambda,n) \, \psi_B(\lambda,n-1) - \psi_B(\lambda,n) \, \psi_A(\lambda,n-1) \Bigr)
\= \lambda \,. 
\eeq
Clearly, the meaning of this identity is the following: both sides of~\eqref{toshowab} 
are Laurent series of~$\lambda^{-1}$ with coefficients in $\widetilde V=V=Q[n]$, and 
the equality means all the coefficients should be equal. 
More precisely, the identity~\eqref{toshowab} can be equivalently written as  
the following sequence of identities:
\beq\label{alreadyrecasts}
\frac{n}{j+1}\sum_{j_1=0}^{j+1} \frac{(-1)^{j_1}}2 \binom{j+1}{j_1} \binom{n+2j_1-1}{2j+1} 
\+ \sum_{j_1=0}^j (-1)^{j_1} \binom{j}{j_1} \binom{n+2j_1}{2j+1} \= 0\,,\quad j\geq 0\,.
\eeq
From~\eqref{alreadyprovedsd} we know that the left-hand side of~\eqref{alreadyrecasts}  
as a polynomial of~$n$ is a constant for any $j\geq 0$. 
Note that the value of the left-hand side of~\eqref{alreadyrecasts}  
at $n=0$ is obviously~$0$ for any $j\geq 0$. 
The proposition is proved.
\end{proof}
It follows from the above proposition an 
explicit expression for the $D(\lambda,\mu,n,\bdzero)$
 (cf.~equation~\eqref{defineD}) associated to the pair~\eqref{pair11gue}--\eqref{pair12gue}:
\beq\label{DGUEexpress}
\frac{e^{s(n-1)}}\mu \, D(\lambda,\mu,n,\bdzero) \, 
\biggl(\frac{\mu}{\lambda}\biggr)^n \= \frac{1}{\lambda-\mu} 
\+ A(\lambda,\mu,n)\,,
\eeq
with $A(\lambda,\mu,n)$ given by
\beq
A(\lambda,\mu,n) \= \sum_{k\geq1} \frac{(2k-1)!!}{(2k)!}\sum_{p=0}^{2k-1}
(-1)^{p+[(p+1)/2]} \binom{k-1}{[p/2]} \prod_{j=-p}^{2k-1-p} (n+j) \, \lambda^{-p-1} \mu^{-(2k-p)} \,. \label{Zhouseries}
\eeq
This explicit expression~\eqref{Zhouseries} first appeared in~\cite{Zhou3}. 
Denote 
\beq\label{ZhouseriesAGUE}
\widehat A(\lambda,\mu,n)\=\frac{1}{\lambda-\mu} 
\+ A(\lambda,\mu,n)\,.
\eeq
As a corollary of Proposition~\ref{GUEpsi}, 
Theorem~\ref{main1} and the above~\eqref{DGUEexpress} we have achieved 
a new proof of the following theorem of Jian Zhou. 
\begin{theorem}[Zhou, \cite{Zhou3}]
Fix $k\geq 2$ being an integer. 
The generating series of $k$-point connected GUE correlators has the following expression:
\begin{align}
& \sum_{i_1,\dots,i_k\geq 1} 
\frac{\bigl\langle{\rm tr} \, M^{i_1}\cdots {\rm tr} \, M^{i_k}\bigr\rangle_{\rm c}}{\lambda_1^{i_1+1} \cdots \lambda_k^{i_k+1}} 
\=  
(-1)^{k-1}  \sum_{\pi \in \cS_k/C_k} 
\prod_{j=1}^k \widehat A \, \bigl(\lambda_{\pi(j)},\lambda_{\pi(j+1)},n\bigr) 
\,-\, \frac{ \delta_{k,2}}{(\lambda_1-\lambda_2)^2} \,,  \label{mainidentitygue}
\end{align}
where $\widehat A$ is defined by~\eqref{Zhouseries}--\eqref{ZhouseriesAGUE}.
Here we recall that for any fixed $i_1,\dots,i_k$, the connected GUE correlator 
$\langle{\rm tr} \, M^{i_1}\cdots {\rm tr} \, M^{i_k}\rangle_{\rm c}$ is 
a polynomial of~$n$ (cf. \cite{BIZ,HZ,KKN,DuY1}).
\end{theorem}

\subsection{Application II. Gromov--Witten invariants of~$\mathbb{P}^1$ in the stationary sector.}
The initial data for the Gromov--Witten solution to the Toda lattice hierarchy 
was for example derived in~\cite{DuY1,DYZ1,DuY2}. It  
has the following explicit expression:
\beq
f(n)\=n\e+\frac\e2\,, \quad  g(n)\=1\,.
\eeq
We have 
$$s(n) \= -  \bigl(1-\Lambda^{-1}\bigr)^{-1} \log 1 \= C\,,$$
where $C$ is an arbitrary constant. Below we take $C=0$. 
\begin{prop}
The $\psi_1,\psi_2$ defined by 
\begin{align}
&\psi_A(\lambda,n) \= \e^{\frac\lambda\e-\frac12} \, \Gamma\Bigl(\frac\lambda\e+\frac12\Bigr)\, J_{\frac\lambda\e-n-\frac12}\Bigl(\frac2\e\Bigr) \,, \label{bessel1}\\
&\psi_B(\lambda,n) \=  (-1)^{n+1} \e^{-\frac\lambda\e-\frac12} \,  \lambda\, 
\Gamma\Bigl(-\frac\lambda\e+\frac12\Bigr)\, J_{-\frac\lambda\e+n+\frac12}\Bigl(\frac2\e\Bigr) \label{bessel2}
\end{align}
form a particular pair of wave functions associated to 
$f(n)=n\e+\frac\e2,g(n)=1$.
Here, $J_\nu(y)$ denotes the Bessel function, and the right-hand sides of~\eqref{bessel1}--\eqref{bessel2} are understood as the 
large $\lambda$ asymptotics of the corresponding 
analytic functions.
\end{prop}
\begin{proof} 
Firstly, using the properties of Bessel functions we can verify that 
$\psi_A(\lambda,n)$ and $\psi_B(\lambda,n)$ defined from the above asymptotics satisfy 
\begin{align}
& \psi_A(\lambda,n+1)\+ \Bigl(n\e+\frac\e2\Bigr) \, \psi_A(\lambda,n) \+  \psi_A(\lambda,n-1) 
\= \lambda \, \psi_A(\lambda,n) \,,\nn\\
& \psi_B(\lambda,n+1)\+ \Bigl(n\e+\frac\e2\Bigr) \, \psi_B(\lambda,n) \+  \psi_B(\lambda,n-1) 
\= \lambda \, \psi_B(\lambda,n) \,. \nn
\end{align}
Secondly, as $\lambda$ goes to~$\infty$, the following asymptotics hold true:
\begin{align}
&\e^{\frac\lambda\e-\frac12} \, 
\Gamma\Bigl(\frac\lambda\e+\frac12\Bigr)\, J_{\frac\lambda\e-n-\frac12}\Bigl(\frac2\e\Bigr) 
\;\sim\;  
\lambda^n \Bigl(1\+ {\rm O}\bigl(\lambda^{-1}\bigr)\Bigr)\,,\nn\\
&  (-1)^{n+1} \e^{-\frac\lambda\e-\frac12} \,  \lambda\, 
\Gamma\Bigl(-\frac\lambda\e+\frac12\Bigr)\, J_{-\frac\lambda\e+n+\frac12}\Bigl(\frac2\e\Bigr)
\;\sim\;  
\lambda^{-n} \Bigl(1\+ {\rm O}\bigl(\lambda^{-1}\bigr)\Bigr) \,. \nn
\end{align}
Thirdly, $\psi_A$ and $\psi_B$ also satisfy 
$$\psi_A(\lambda,n) \, \psi_B(\lambda,n-1) \,-\, \psi_B(\lambda,n) \, \psi_A(\lambda,n-1) \= \lambda\,.$$
We have verified all the defining properties for a pair of wave functions associated to $f(n)=n\e+\frac\e2,g(n)=1$. 
The proposition is proved.
\end{proof}

Note that 
\begin{align}
&\psi_A(\lambda,n-1)\= \e^{\frac\lambda\e-\frac12} \, \Gamma\Bigl(\frac\lambda\e+\frac12\Bigr)\, 
J_{\frac\lambda\e-n+\frac12}\Bigl(\frac2\e\Bigr) \,, \label{bessel3}\\
&\psi_B(\lambda,n-1)\= (-1)^{n} \e^{-\frac\lambda\e-\frac12} \,  \lambda\, 
\Gamma\Bigl(-\frac\lambda\e+\frac12\Bigr)\, J_{-\frac\lambda\e+n-\frac12}\Bigl(\frac2\e\Bigr) \,,\label{bessel4}
\end{align}
and denote 
$$
J_\nu(y) \,=:\, \frac{(y/2)^\nu}{\Gamma(\nu+1)}  j_{\nu+\frac12}(y^2/4) \,.
$$
It follows from~\eqref{bessel1}--\eqref{bessel4},~\eqref{defineD} that the 
$D(\lambda,\mu,0,\bdzero)$
associated to the pair~\eqref{bessel1}--\eqref{bessel2} has the following explicit expression:
$$\frac1\mu D(\lambda,\mu,0,\bdzero) \= - \frac1\e \,
\frac{j_{-\frac\mu\e}\bigl(\frac1{\e^2}\bigr) \, j_{\frac\lambda\e}\bigl(\frac1{\e^2}\bigr)  
  + 
 \frac{\e^{-2}}{(\frac12-\frac\mu\e)(\frac12+\frac\lambda\e)} \, 
 j_{1-\frac\mu\e}\bigl(\frac1{\e^2}\bigr)\, 
 j_{1+\frac\lambda\e}\bigl(\frac1{\e^2}\bigr)}{\mu/\e-\lambda/\e} \,.$$
Then according to~\cite{DYZ1}, the function $\frac1\mu D(\lambda,\mu,0,\bdzero)$ 
has the following expressions:
\begin{align}
& \frac1\mu D(\lambda,\mu,0,\bdzero) \nn\\
& \quad \= 
-\frac1\e \sum_{k=0}^\infty \frac{(a-b-2k+1)_{k-1}}{k! \, (-a+\tfrac12)_k \, (b+\tfrac12)_k} \e^{-2k}  \label{Dhgm}\\
& \quad \=  
\frac{-1}{\e(a-b)} ~  {}_2 F_3\Bigl(\frac{b-a}2, \frac{b-a+1}2; \, \frac12 -a \, , \, \frac12 + b \,, b-a+1; \, -4 \e^{-2}\Bigr) \\
& \quad \, \sim \;  
 \frac{-1}{\e(a-b)}  \,-\, \sum_{p,q\geq 0} \frac{(-1)^{q+1}}{a^{p+1} b^{q+1}} 
\sum_{k\geq 1} \frac{\e^{-2k-1}}{k!} \nn\\
& \qquad \qquad \qquad  \sum_{1\leq i,j\leq k} (-1)^{i+j}\frac{(i+j-2k)_{k-1} \bigl(i-\frac12\bigr)^p \bigl(j-\frac12\bigr)^q}{(i-1)!(j-1)!(k-i)!(k-j)!} \;=:\; \widehat A \, (\lambda,\mu) \,, \label{asymD}
\end{align}
where $a:=\frac\mu\e$, $b:=\frac\lambda\e$, the $(a-b+1)_{-1}$ of~\eqref{Dhgm} is defined as $1/(a-b)$, 
and $\sim$ in~\eqref{asymD} is taken as $a,b\rightarrow \infty$ away from the half integers. 
The explicit expression~\eqref{asymD} first appeared in~\cite{DYZ1}.
So we have completed a new proof of the following theorem.
\begin{theorem}[\cite{DYZ1}]
The generating series of $k$-point ($k\geq 2$) Gromov--Witten invariants of~$\mathbb{P}^1$ 
in the stationary sector has the following explicit expression:
\begin{align}
& \e^k \! \sum_{i_1,\dots,i_k\geq 0}  
\frac{(i_1+1)! \cdots (i_k+1)!}{\lambda_1^{i_1+2} \cdots \lambda_k^{i_k+2}} 
\langle  \tau_{i_1}(\omega) \cdots \tau_{i_k}(\omega)  \rangle(\e) \nn\\
& \quad  \= (-1)^{k-1} \sum_{\pi \in \cS_k/C_k} 
\prod_{i=1}^k \widehat A \, \bigl(\lambda_{\pi(i)}, \lambda_{\pi(i+1)}\bigr)  \,-\, \frac{\delta_{k,2}}{(\lambda_1-\lambda_2)^2} \,, 
\end{align}
where $\widehat A \, (\lambda,\mu)$ is explicitly defined in~\eqref{asymD}, and 
\beq \langle\tau_{i_1}(\omega) \cdots \tau_{i_k}(\omega)  \rangle(\e) \:= \sum_{g\geq 0} 
\e^{2g-2}  \sum_{d\geq 0}
\int_{\left[\overline{\mathcal{M}}_{g,k}(\mathbb{P}^1,d)\right]^{{\rm virt}}} 
{\rm ev}_1^*(\omega)  \cdots  {\rm ev}_k^*(\omega) \, \psi_1^{i_1} \cdots \psi_k^{i_k} \,.  \label{GWp1}
\eeq 
(See for example~\cite{DYZ1} for the notation about the integral in the right-hand side of~\eqref{GWp1}.)
\end{theorem}

\begin{appendix}
\section{Pair of abstract pre-wave functions}
Here we construct a ring that is suitable for defining abstract pre-wave functions.
Recall that $\A$ is the ring of polynomials of $v_k,w_k$, $k \in \ZZ$. Instead of the $\ZZ$-coefficients, 
we will use in this appendix the $\QQ$-coefficients, 
i.e., $\A=\QQ \bigl[ \{v_k,w_k \,|\, k\in \ZZ\}\bigr]$ is now under consideration. 
For each monic monomial~$\alpha\in\A\backslash\QQ$, 
we associate a symbol~$m_\alpha$. Denote by~$\cB$ the polynomial ring 
\beq \cB\:=\QQ\bigl[\{ \,m_\alpha\,|\,\alpha \mbox{ is a monic monomial in } \A\backslash\QQ \} \bigr]\,. \eeq
Define the action of~$\Lambda^k$ on~$\cB$ with $k\in \ZZ$ by
\beq
\Lambda^k \bigl(m_{\alpha_1} \cdots m_{\alpha_l}) \= m_{ \Lambda^k(\alpha_1)} \cdots m_{ \Lambda^k (\alpha_l)} 
\eeq
for $\alpha_1,\dots,\alpha_l$ being monic monomials in~$\A\backslash\QQ$, 
as well as by linearly extending it to other elements of~$\cB$. 
For a monic monomial $\alpha=v_{i_1} \cdots v_{i_r} w_{j_1} \cdots w_{j_s}\in \A\backslash\QQ$ with $i_1\leq\dots\leq i_r$, $j_1\leq\dots\leq j_s$ and 
$r+s\geq 1$, let $k_\alpha:=-i_1$ (if $r\geq 1$), $k_\alpha:=-j_1$ (if $r=0$); 
the monomial $\Lambda^{k_\alpha} (\alpha)\in \A$ is then called the 
(unique) reduced monomial (associated to~$\alpha$). 
Denote by~$\mathcal{C}$ the polynomial ring generated by 
$m_{\beta}$, $v_k$, $w_k$ with $\QQ$-coefficients, where 
$\beta$ are reduced monic monomials, and $k\in \ZZ$.
Let us also define an action of~$\Lambda^k$ on~$\mathcal{C}$, $k\in\ZZ$. To this end, we introduce some notations:
for $\beta$ a reduced monic monomial of~$\A$, denote 
\beq
n_{\Lambda^k(\beta)}  \:= \left\{
\begin{array}{cc} 
m_\beta + \sum_{i=0}^{k-1} \Lambda^i(\beta)\,, & ~ k\geq 0\,,\\   
m_\beta - \sum_{i=k}^{-1} \Lambda^i(\beta)\,, & ~ k\leq -1\,.
\end{array}\right.
\eeq
Then for a 
monomial $\alpha \cdot m_{\beta_1} \cdots m_{\beta_s}$ of~$\mathcal{C}$ with $\alpha$ being a monomial in~$\A$, define
\beq
\Lambda^k (\alpha \cdot m_{\beta_1} \cdots m_{\beta_s}) = \Lambda^k (\alpha)  \cdot \prod_{j=1}^s n_{\Lambda^k(\beta_j)}\,, \quad k\in\ZZ\,.
\eeq
Define the action of~$\Lambda^k$ on other elements in~$\mathcal{C}$ by requiring it as a linear operator.
Denote by $p:\mathcal{B}\rightarrow \mathcal{C}$ the linear map which maps  
$m_{\alpha_1} \cdots m_{\alpha_l}\in \cB$ to $n_{\alpha_1} \cdots n_{\alpha_l}\in \mathcal{C}$, 
for~$\alpha_i$, $i=1,\dots,l$ being monic monomials in~$\A\backslash\QQ$.
Denote by $\cB^0$ the image of~$p$. 
Clearly, $\A\subset \cB^0$. Indeed, the 
 element $(\Lambda-1) \bigl(\sum_{i=1}^l \lambda_i \, m_{\alpha_i}\bigr) \in \cB$ 
 becomes $\sum_{i=1}^l \lambda_i \alpha_i \in \A$ 
 under the map~$p$. Here $\alpha_1,\dots,\alpha_l$ are distinct monic monomials in~$\A\backslash\QQ$.
Finally we define an 
operator $\mathbb{S}: \A\backslash\QQ \rightarrow \cB^0$ by 
\beq
\mathbb{S} \bigl(\lambda_1 \alpha_1 \+ \cdots \+ \lambda_l \alpha_l\bigr) \= \lambda_1 n_{\alpha_1} \+ \cdots \+ \lambda_l n_{\alpha_l} \,
\eeq
for $\alpha_1,\dots,\alpha_l$ being distinct monic monomials and $\lambda_1,\dots,\lambda_l\in\QQ$.

Motivated by \eqref{ykn} and~\eqref{zkn}, define two families of elements $y_i,z_i \in \A$, $i\geq 1$ by 
\begin{align}
& y_{k+1} \= 
- \sum_{m_1,\dots,m_k\geq 0\atop \sum_{i=1}^{k} im_i=k+1} 
\frac{\prod_{i=1}^k y_i^{m_i}}{\prod_{i=1}^k m_i!} 
-v_0\delta_{k,0}  - w_0 \sum_{m_1,\dots,m_{k-1}\geq 0\atop \sum_{i=1}^{k-1} im_i=k-1} 
\frac{\prod_{i=1}^{k-1} (-1)^{m_i} \bigl(\Lambda^{-1}(y_i)\bigr)^{m_i}}{\prod_{i=1}^{k-1} m_i!} \,, \nn \\
& z_{k+1} \= 
\sum_{m_1,\dots,m_k\geq 0\atop \sum_{i=1}^{k} im_i=k+1} 
\frac{\prod_{i=1}^k (-1)^{m_i} z_i^{m_i}}{\prod_{i=1}^k m_i!} 
\+ v_1\delta_{k,0}  \+ w_2 \sum_{m_1,\dots,m_{k-1}\geq 0\atop \sum_{i=1}^{k-1} im_i=k-1} 
\frac{\prod_{i=1}^{k-1} \bigl(\Lambda(z_i)\bigr)^{m_i}}{\prod_{i=1}^{k-1} m_i!} \,. \nn
\end{align}
Equivalently, the generating series $y(\lambda):=\sum_{i\geq1} y_i/\lambda^i$, $z(\lambda):=\sum_{i\geq1} z_i/\lambda^i$ satisfy
\begin{align}
& \lambda \, e^{y(\lambda)} \+ v_0 \,-\, \lambda \+ w_0 \, \lambda^{-1} \Lambda^{-1}\bigl(e^{-y(\lambda)}\bigr) \= 0\,, \nn \\
& \lambda \, \Lambda^{-1} \bigl(e^{-z(\lambda)}\bigr) \+ v_0 \,-\, \lambda \+ w_1 \, \lambda^{-1} e^{z(\lambda)} \= 0\,. \nn
\end{align}
Define 
\beq
\psi_A\:= e^{\mathbb{S} (y(\lambda))} \otimes \lambda^n \otimes 1\,, 
\quad \psi_B \:= e^{\mathbb{S} (z(\lambda))} \otimes \lambda^{-n} \otimes e^{-\sigma}  \,,
\eeq
where $e^{-\sigma}$ is a formal element satisfying $e^{(1-\Lambda^{-1})(-\sigma)} = w_0$, and $\lambda^n$, $\lambda^{-n}$ are  
formal elements satisfying $\Lambda^k (1\otimes\lambda^n)=\lambda^k \otimes \lambda^n$, 
$\Lambda^k (1\otimes\lambda^{-n})=\lambda^{-k} \otimes \lambda^{-n}$, $k\in \ZZ$.
We have 
\begin{align}
& L \bigl(\psi_A\bigr) \= \lambda \, \psi_A\,,\quad L \bigl(\psi_B\bigr) \= \lambda \, \psi_B \,,\nn\\
& \psi_A(\lambda) \= \bigl(1+{\rm O}(\lambda^{-1})\bigr)\otimes \lambda^n
\;\in\; \mathcal{C} \bllm\lambda^{-1}\brrm \otimes \lambda^n \,,\\
& \psi_B(\lambda) \= \bigl(1+{\rm O}(\lambda^{-1})\bigr) \otimes \lambda^{-n} \otimes e^{-\sigma}   \;\in\; 
\mathcal{C} \bllm \lambda^{-1}\brrm  \otimes \lambda^{-n}  \otimes e^{-\sigma} \,,
\end{align}
where $L=\Lambda+v_0 + w_0 \, \Lambda^{-1}$. 
We call $\psi_A$ and $\psi_B$ the abstract pre-wave functions of type~A and of type~B, respectively, associated to $v_0,w_0$.

Denote 
\beq\label{wavepairddpre}
d_{\rm pre}(\lambda)\:= \psi_A(\lambda) \, \Lambda^{-1} \bigl(\psi_B(\lambda)\bigr) \,-\, \psi_B(\lambda) \, \Lambda^{-1}\bigl(\psi_A(\lambda)\bigr) 
\eeq
and
\beq\label{Psipredefinition}
\Psi(\lambda) \:= \begin{pmatrix} 
\psi_A (\lambda) & \psi_B(\lambda) \\ 
\Lambda^{-1} \bigl(\psi_A(\lambda)\bigr) & \Lambda^{-1} \bigl(\psi_B(\lambda)\bigr) \end{pmatrix}\,. 
\eeq
Then we have the following identity:
\beq\label{RPpre}
R(\lambda) \= \Psi(\lambda) \, \begin{pmatrix} 1 & 0 \\ 0 & 0 \end{pmatrix} 
\, \Psi^{-1}(\lambda)  \;=:\;  M(\lambda) \,.
\eeq
The proof is similar to that of Proposition~\ref{propRP}. (The main fact used in the proof is that from the definition the 
coefficients of entries of~$R(\lambda)$ are uniquely determined in an algebraic way.)
We omit its details here. However, let us explain the equality~\eqref{RPpre} by an equivalent form. 
From definition we have 
\begin{align}
M(\lambda) \= 
\frac1{d_{\rm pre}(\lambda)} 
\begin{pmatrix} 
\psi_A (\lambda) \, \Lambda^{-1} \bigl(\psi_B(\lambda)\bigr) & -\psi_A (\lambda) \, \psi_B(\lambda) \\ 
\Lambda^{-1} \bigl(\psi_A(\lambda)\bigr) \, \Lambda^{-1} \bigl(\psi_B(\lambda)\bigr)  & -\Lambda^{-1} \bigl(\psi_A(\lambda)\bigr) \, \psi_B(\lambda) \end{pmatrix}\,.  \nn
\end{align}
Then from a straightforward calculation by using the definitions we find
\begin{align}
M_{11}(\lambda) & 
\= \frac{1}{1 \,-\, \frac{w_0}{\lambda^2} \, e^{ \Lambda^{-1} (z(\lambda) - y(\lambda))}} \,, \label{m11yz}\\
M_{12}(\lambda) & 
\= \frac{1}{\lambda^{-1}  \, e^{-\Lambda^{-1}(y(\lambda))}  
\,-\, \frac{\lambda}{w_0} \, e^{-\Lambda^{-1}(z(\lambda))}} \,, \label{m12yz}\\
M_{21}(\lambda) & 
\= \frac{1}{\lambda  \, e^{\Lambda^{-1}(y(\lambda))}  \,-\, \frac{w_0}{\lambda} \, e^{\Lambda^{-1}(z(\lambda))}}\,, \label{m21yz}\\
M_{22}(\lambda) & 
\= \frac{1}{1 \,-\, \frac{\lambda^2}{w_0} \, e^{ \Lambda^{-1} (y(\lambda)- z(\lambda))}} \,. \label{m22yz}
\end{align}
Hence the equality~\eqref{RPpre} means new expressions for the entries of the basic matrix resolvent~$R(\lambda)$ 
explicitly in terms of $y(\lambda),z(\lambda)$.
Substituting the following expansions  
\beq
y(\lambda) \= -\frac{v_0}{\lambda } \,-\, \frac{\frac{1}{2} v_0^2+w_0}{\lambda^2} \+ \cdots \,, \qquad 
z(\lambda) \= \frac{v_1}{\lambda } \+ \frac{\frac{1}{2} v_1^2+w_2}{\lambda^2} \+ \cdots
\eeq
into~\eqref{m11yz}--\eqref{m22yz} we find that the new expressions agree with~\eqref{Rfirstfew}.
Combining with~\eqref{taualphabeta}--\eqref{taugamma} we obtain
\begin{align}
& \frac{1}{\frac{\lambda^2}{w_0} \, e^{ \Lambda^{-1} (y(\lambda)- z(\lambda))} \,-\, 1} \= \sum_{p\geq 0} \Omega_{p,0} \, \lambda^{-p-2} \;=:\; A \,, \\
& \frac{1}{\lambda  \, e^{\Lambda^{-1}(y(\lambda))}  \,-\, \frac{w_0}{\lambda} \, e^{\Lambda^{-1}(z(\lambda))}} \= \lambda^{-1}\+\sum_{p\geq 0} \Lambda^{-1} \bigl(S_p\bigr) \, \lambda^{-p-2} \; =: \;  B\,.
\end{align}
We therefore arrive at the following formulae:
\beq
e^{\Lambda^{-1}(y(\lambda))} \= \frac1\lambda\, \frac{1+A}{B} \,, \qquad 
e^{\Lambda^{-1}(z(\lambda))} \= \frac{\lambda}{w_0} \, \frac{A}{B} \,.
\eeq

Let us proceed to the generating series of multi-point correlation functions. Define 
\beq\label{defineDpre}
D_{\rm pre}(\lambda,\mu) \:= \frac{\psi_A(\lambda) \, \Lambda^{-1} \bigl(\psi_B(\mu)\bigr) \,-\, \Lambda^{-1} \bigl(\psi_A(\lambda)\bigr) \,\psi_B(\mu)}{\lambda-\mu} \,.
\eeq
Using~\eqref{RPpre}, Proposition~\ref{matrixprop} and a similar argument to the proof of Theorem~\ref{main1time} we obtain
\begin{align}
& \sum_{i_1,\dots,i_k\geq 0} 
\frac{\Omega_{i_1,\dots,i_k}}{\lambda_1^{i_1+2} \cdots \lambda_k^{i_k+2}}  \=  
 \frac{(-1)^{k-1}}{\prod_{j=1}^k d_{\rm pre}(\lambda_j)}  
\sum_{\pi \in \cS_k/C_k} \prod_{j=1}^k D_{\rm pre}(\lambda_{\pi(j)},\lambda_{\pi(j+1)}) 
\,-\, \frac{\delta_{k,2}}{(\lambda_1-\lambda_2)^2}  \,.  \label{mainidentitypre}
\end{align}
For the reader's convenience, we give the first few terms of the abstract pre-wave functions $\psi_A(\lambda)$ and~$\psi_B(\lambda)$ as follows:
\begin{align}
& \psi_A \=  \biggl(1 \,-\, \frac{m_{v_0}}{\lambda}  \+ \frac{m_{v_0}^2-m_{v_0^2}-2m_{w_0}}{2 \lambda^2}  \nn\\
& \qquad\qquad - \frac1{6\lambda^3} \Bigl(m_{v_0}^3 + 2m_{v_0^3} - 3m_{v_0} m_{v_0^2} 
+ 6m_{v_0 w_0} + 6m_{v_0 w_1} \nn\\
& \qquad \qquad \qquad \qquad - 6 m_{v_0} m_{w_0} - 6v_{-1} w_0\Bigr) \+ {\rm O}\Bigl(\frac1{\lambda^4}\Bigr)\biggr) \, \lambda^n \,,\\
& \psi_B \=  \biggl(1 \+\frac{m_{v_0}+v_0}{\lambda }  \+ \frac{m_{v_0}^2+m_{v_0^2} +2v_0 m_{v_0}+2m_{w_0}+2v_0^2+2w_0+2w_1}{2\lambda^2} \nn\\
& \qquad \qquad \+ \frac1{6 \lambda^3} \Bigl(m_{v_0}^3+6m_{v_0} m_{w_0} +3 m_{v_0} m_{v_0^2}  +2m_{v_0^3} 
+6m_{v_0 w_1} 
+ 6m_{v_0 w_0} \nn\\
& \qquad\qquad \qquad \qquad  +3 v_0 m_{v_0}^2+6v_0^2m_{v_0}  +6w_0 m_{v_0} +6w_1 m_{v_0} +3v_0 m_{v_0^2}+6v_0 m_{w_0} \nn\\
& \qquad\qquad \qquad \qquad +6 v_0^3+12v_0 w_0+12 v_0 w_1 + 6 v_1 w_1 \Bigr) \+ {\rm O}\Bigl(\frac1{\lambda^4}\Bigr)\biggr) \, \lambda^{-n} e^{-\sigma} \,.
\end{align}
It turns out that the above abstract pre-wave functions form {\it a pair}.
Namely, $d_{\rm pre}(\lambda) = \lambda\, e^{\Lambda^{-1}(-\sigma)}$. 
Interestingly, for given arbitrary initial value $(f(n),g(n))$, 
based on this statement one obtains a constructive method for a pair of wave functions 
associated to $(f(n),g(n))$ (cf.~\eqref{normalization} in Section~\ref{Introwave} for the definition of a pair). 
This is important considering Theorem~\ref{main1}.
We hope to confirm the statement on the pair property of the abstract pre-wave functions in another publication. 
\end{appendix}

\bigskip
\medskip

\noindent School of Mathematical Sciences,  University of Science and Technology of China

\noindent Hefei 230026, P.R. China 

\noindent diyang@ustc.edu.cn

\end{document}